\documentclass{article}

\usepackage{fullpage}

\usepackage{cite}
\usepackage{amsmath,amssymb,amsfonts}
\usepackage{algorithmic}
\usepackage{graphicx}
\usepackage{textcomp}
\usepackage{xcolor}
\usepackage{tikz}
\usetikzlibrary{decorations.pathmorphing}
\usepackage[hypertexnames=false]{hyperref}

\usepackage{amsthm}
\usepackage{comment}

\newtheorem{lemma}{Lemma} 
\newtheorem{proposition}{Proposition} 
\newtheorem{theorem}{Theorem} 
\newtheorem{corollary}{Corollary} 
\newtheorem{claim}{Claim} 
\newtheorem*{positive-result*}{Positive result} 
\newtheorem*{negative-result*}{Negative result} 

\theoremstyle{definition}
\newtheorem{definition}{Definition}   
  
\newtheorem{remark}{Remark}   
\newtheorem{question}{Question}    

\newtheorem*{notation_nonumber}{Notation}

\newcommand{\eqp}{\stackrel{\log}{=}}
\newcommand{\gep}{\stackrel{\log}{\ge}}
\newcommand{\lep}{\stackrel{\log}{\le}}

\renewcommand{\vert}{\!\mid\!}

\providecommand{\keywords}[1]{\textbf{\textit{Keywords:}} #1}

\def\firstcircle{(150:1.75cm) circle (2.5cm)}
\def\secondcircle{(30:1.75cm) circle (2.5cm)}
\def\thirdcircle{(270:1.75cm) circle (2.5cm)}

\title{Algebraic Barriers to Halving Algorithmic Information Quantities in Correlated Strings\footnote{%
A short version of this paper has been accepted to the  MFCS 2025 conference.
The extended version contains a more detailed introduction, 
a new Section~\ref{sec:no-halving-raz} with a proof based on Razenshteyn's theorem,
and a discussion of different approaches to the proof of  Theorem~\ref{th:halving-log}  in Appendix.
}}

\author{Andrei Romashchenko}

\begin{document}

\maketitle

\begin{abstract}
We study the possibility of scaling down algorithmic information quantities in tuples of correlated strings. 
In particular, we address a question raised by Alexander Shen: whether, for any triple of strings \((a, b, c)\), 
there exists a string \(z\) such that each conditional Kolmogorov complexity \(C(a|z), C(b|z), C(c|z)\) is approximately half of the corresponding unconditional Kolmogorov complexity. 
We give a negative answer to this question by constructing a triple \((a, b, c)\) for which no such string \(z\) exists. 
Moreover, we construct a fully explicit example of such a tuple.
Our construction is based on combinatorial properties of incidences in finite projective planes and relies on  bounds for point-line incidences over prime fields.
As an application, we show that this impossibility yields lower bounds on the communication complexity of secret key agreement protocols in certain settings. 
These results reveal algebraic obstructions to efficient information exchange and highlight a separation in information-theoretic behavior between 
fields with and without proper subfields.
\end{abstract}


\newpage

\section{Introduction}
Algorithmic information theory (AIT), introduced and developed  in the 1960s by 
Solomonoff \cite{solomonoff1,solomonoff2,solomonoff3}, Kolmogorov~\cite{kolmogorov1965}, and Chaitin~\cite{chaitin1969},  
aims to define the amount of information in a discrete object and to quantify the information shared between several objects.
The crucial difference from Shannon's information theory is that AIT is  interested not in an \emph{average} compression rate 
(for a given distribution of probabilities) but in the optimal compression of some specific \emph{individual} object.
Informally, the information content of an individual object (for example, a string or a text)
 is defined as the minimal length of a program that produces that object. 
The length of the shortest program producing a string $x$ 
is called the Kolmogorov complexity of $x$ and denoted $C(x)$.
Similarly, the conditional Kolmogorov complexity of $x$ given $y$, denoted $C(x\vert y)$,
is the  length of an optimal program producing a string $x$ given $y$ as input.
Note that $C(x\vert y)$ makes sense even for infinite $y$ 
(this quantity can be interpreted as the length of the shortest program that produces a string $x$ while having access to the oracle $y$). 
A string $x$ is called random or incompressible if $C(x) \approx |x|$.
The value of Kolmogorov complexity depends on the chosen programming language. 
However,  it is known that there exist optimal programming languages making the complexity function minimal up to a bounded additive term. 

AIT is tightly connected with classical Shannon information theory. 
The techniques of Kolmogorov complexity are used in various problems of theoretical computer science and discrete mathematics. 
Time-bounded Kolmogorov complexity has deep links with computational complexity and theoretical cryptography,
see, e.g., the surveys \cite{survey-fortnow} and \cite{survey-oliveira}.
An extensive introduction to AIT and the theory of Kolmogorov complexity can be found, for example, in the classical paper \cite{zl},
in the textbooks \cite{li-vitanyi,suv} or in the surveys \cite{vyugin-survey1,vyugin-survey2,vyugin-survey3}.

\smallskip

One of the fundamental questions of AIT is the characterization of  the possible values of Kolmogorov complexity for tuples of strings. For example, for any triple of strings $x_1,x_2,x_3$,
we have seven values of Kolmogorov complexity (sometimes called \emph{complexity profile} of $(x_1,x_2,x_3)$): 
\[
C(x_1),\ C(x_2), \ C(x_3),\ C(x_1,x_2),\ C(x_1,x_3), \ C(x_2,x_3), \ C(x_1,x_2,x_3).
\]
Which vectors of seven positive numbers can be realized as the vector of Kolmogorov complexities of some $x_1,x_2,x_3$? 
When $n$ strings $x_i$ (for $i=1,\ldots,n$) are involved, we may consider Kolmogorov complexities of all tuples $C(x_{i_1},\ldots, x_{i_s})$ for all selections of indices $1\le i_1 < i_2 < \ldots <i_s \le n$,
and the question is: which vectors of $2^{n}-1$ positive numbers can be realized as the Kolmogorov complexity of some $(x_1,\ldots,x_n)$.
This problem appears to be more combinatorial than algorithmic in nature. 
The answer to this question does not depend significantly on the choice of the optimal programming language.
Moreover, even relativization of Kolmogorov complexity with respect to an oracle 
(when we replace the plain Kolmogorov complexity by Kolmogorov complexity in the sense of programs with access to some fixed oracle)
would not change the answer to this question significantly, as the following proposition shows.
%todo
\begin{proposition}[folklore]%  reference? folklore ?
\label{p:unrelativize}
For every tuple of finite binary strings $(y_1,\ldots, y_n)$ and for every $z$ (finite or infinite),  
there exists a tuple $(x_1,\ldots, x_n)$ such that for all index sets  
$1\le i_1 < i_2 < \ldots < i_s \le n$,
\[
C(x_{i_1},x_{i_2},\ldots,x_{i_s})
= C(y_{i_1},y_{i_2},\ldots,y_{i_s} \mid z)
\pm O(\log C(y_1,\ldots, y_n)).
\]
In other words, if a vector of $2^n-1$ reals serves as the complexity profile for some $(y_1,\ldots, y_n)$ conditional on some oracle $z$,
then a vector of almost the same numbers serves as  the complexity profile for some other strings $(x_1,\ldots, x_n)$ for the values of their plain Kolmogorov complexities (without any oracle).
\end{proposition}
For completeness,  we prove this proposition in Appendix~\ref{s:appendixB}.

\smallskip

Thus, characterizing the class of all possible complexity profiles is a natural question. 
What are the constraints linking different components of such complexity vectors?
Some of these constraints are simple and standard. 
There are, for example, classical inequalities
\[
C(a) \le C(a,b)  + O(1),
\]
(monotonicity)
\[
C(a,b) \le C(a) + C(b) + O(\log C(a,b))
\]
(subadditivity, or non-negativity of the mutual information) and 
\[
C(a,b,c) + C(c)\le C(a,c) + C(b,c) + O(\log C(a,b,c))
\]
(submodularity, or non-negativity of the conditional mutual information).
We can substitute any tuples of strings for $x_i$ in these inequalities.
For instance, a possible instantiation of subadditivity is 
\[
C(x_1,x_2,x_3,x_4) \le C(x_1,x_2) + C(x_3,x_4) + O(\log C(x_1,x_2,x_3,x_4))
\]
(here we substituted for $a$  and $b$ the tuples $(x_1,x_2)$ and $(x_3,x_4)$ respectively).
It is known that for triples of strings, no other linear inequalities for Kolmogorov complexity exist that differ substantially
from the inequalities mentioned above.
More precisely, any  linear inequality for Kolmogorov complexity of  $x_1,x_2,x_3$ (valid up to an additive term $o(C(x_1,x_2,x_3))$)
is necessarily a positive linear combination of several instances of monotonicity, subadditivity, and submodularity, \cite{hammer2000}. 
However, when four or more strings are involved (so we have $\ge 2^4-1 = 15$ quantities of Kolmogorov complexity), 
there also exist different   linear inequalities  (usually called \emph{non-Shannon-type} inequalities)
that are less intuitive in form and 
cannot be represented as linear combinations of monotonicity, subadditivity, and submodularity; see e.g. the survey \cite{survey-non-shannon}.
It is known that exactly the same linear inequalities hold for Kolmogorov complexity and for Shannon entropy, but the problem of precise characterization 
of these inequalities for $n\ge 4$ objects remains open.

While the questions on linear inequalities for Kolmogorov complexity and for Shannon's entropy are known to be equivalent, from other perspectives, questions about Kolmogorov complexity appear more difficult than 
similar questions about Shannon's entropy. It is not known, for example, whether complexity profiles can be scaled with any factor $\lambda>0$ 
(even up to a logarithmic additive term).
More specifically, the following question is open:
\begin{question}\label{q:1}
Let $\lambda$ be a positive real number.
Is it true that for every $k$-tuple of strings $(x_1,\ldots, x_k)$ 
there exists another $k$-tuple $(x_1',\ldots, x_k')$ such that 
\[
C(x_{i_1}',x_{i_2}',\ldots,x_{i_s}')  =  \lambda C(x_{i_1},x_{i_2},\ldots,x_{i_s}) + O(\log C(x_1,\ldots, x_k))
\] 
for all tuples of indices $(i_1,\ldots, i_s)$, $1\le i_1 < i_2 < \ldots <i_s \le k$ ?
\end{question}
The answer to this question is known to be positive for $k\le 3$ and any $\lambda$, and for any $k$ and integer $\lambda$. 
For non-integer factors, e.g., for $\lambda=1/2$, the question is open for all $k\ge 4$, see \cite{27open}.
Alexander Shen posed  another question (see a comment to Question~1 in \cite{27open}):
\begin{question}\label{q:2}
Let $\lambda<1$ be a positive real number.
Is it true that for every $k$-tuple of strings $(x_1,\ldots, x_k)$ 
there exists a string $z$  such that 
\[
C(x_i \vert z)  =  \lambda C(x_i) + O(\log C(x_1,\ldots, x_k))
\] 
for $i=1,\ldots, k$ ?
\end{question}
A positive answer to Question~\ref{q:2} would imply a positive answer to Question~\ref{q:1}
(see Corollary~\ref{cor:q2->q1} in Appendix~\ref{s:appendixB}). 
Moreover, Question~\ref{q:2} is interesting in its own right as a special case of the problem of \emph{extending complexity profiles},
which generalizes the classical notions of \emph{common information} (due to Gács–Körner~\cite{gk1973} and Wyner~\cite{wyner}).
The corresponding question can be formulated more generally as follows. 
\begin{question}[informal]
\label{q:3}
For each given $k$-tuple of strings $(x_1,\ldots, x_k)$, 
what can we say about possible values
\[
 \{C(x_{i_1},x_{i_2},\ldots,x_{i_s},z)\}_{1\le i_1 < \ldots <i_s\le k}
\] 
achievable with various strings $z$?
\end{question}

It is known that the answer to Question~\ref{q:2} is positive for $k=1$ and for $k=2$ (see Section~\ref{sec:halving}). 
In this paper we give a negative answer to this question for $k=3$, even for $\lambda=1/2$. 
We show that there exists a triple of strings $(a,b,c)$ such that there is no $z$ which ``halves'' the complexities of each of them,
\begin{equation}
\label{eq:halving}
C(a\vert z) \approx \frac12 C(a),\ C(b \vert z) \approx \frac12 C(b),\ C(c \vert z) \approx \frac12 C(c).
\end{equation}
We construct an example of a triple $(a,b,c)$ for which, for every $z$ satisfying 
\(
C(a \vert z) \approx \frac12 C(a),\ C(b \vert z) \approx \frac12 C(b), 
\)
the value of $C(c\vert z)$ must be significantly smaller than $\frac12 C(c)$. 
We prove this statement in Section~\ref{sec:no-halving}.

\subsection{Existential proof vs. explicit constructions}

The negative answer to Question~\ref{q:2} for a \emph{triple} of strings $(a,b,c)$ can be reformulated as a statement about the information-theoretic properties of a certain \emph{pair} of strings $(x,y)$.
In fact, it suffices to show that for every $n$ there exists a pair of strings $(x,y)$ such that:
\begin{itemize}
\item[(i)] $C(x)\approx 2n$, $C(y)\approx 2n$, $C(x,y)\approx 3n$, and
\item[(ii)] there is no $z$ such that $C(x|z)\approx n$, $C(y|z)\approx n$, and $C(x,y|z)\approx 1.5n$.
\end{itemize}
Then we can define $a := x$, $b := y$, and $c := xy$ (i.e., $c$ is the concatenation of the two strings, or possibly an encoding of the pair $(x,y)$),
and such a triple $(a,b,c)$ yields a negative answer to Question~\ref{q:2}.

The existence of a pair $(x,y)$ with the required properties (i) and (ii) follows from a result on the impossibility of materializing mutual information
(large gap between mutual information and common information), proven in \cite{razenshteyn}, see Section~\ref{sec:no-halving-raz} for details.
Formally speaking, this already resolves Question~\ref{q:2}.
However, the underlying proof is not entirely satisfactory.
The reason is that the argument in~\cite{razenshteyn} is based on the probabilistic method: it is purely existential and does not provide an explicit construction of the required tuple of strings.
More precisely, the argument in \cite{razenshteyn} implies the existence of a set $S_n$ such that:
\begin{itemize}
\item $S_n$ consists of $2^{3n + o(n)}$ elements;
\item the complete list of elements in $S_n$ can be enumerated by a program of size $o(n)$;
\item most triples of strings in $S_n$ satisfy the required properties (i) and (ii).
\end{itemize}
However, the proof in~\cite{razenshteyn} sheds no light on the combinatorial structure of the set $S_n$ or the elements it contains; the argument reveals  no symmetries or other natural properties of these objects.
Moreover, enumerating the elements of $S_n$ requires double-exponential time in $n$ (brute-force search).
Thus, the pair $(x,y)$ obtained via~\cite{razenshteyn} appears as an unnatural object, lacking any interesting mathematical structure.

Should Question~\ref{q:2} be revisited with respect to more natural tuples $(a,b,c)$?
V'yugin argued that objects arising in “real life” are usually \emph{typical} elements of \emph{simple} sets, see, e.g.,  \cite{vyugin-most-stochastic}.
This idea traces back to Kolmogorov's definition of $(\alpha,\beta)$-stochasticity, see \cite{shen-non-stochastic}.
Recall that an object $w$ is called $(\alpha,\beta)$-stochastic if there exists a set $W$ such that:
\begin{itemize}
\item there is a program of length at most $\alpha$ that prints the list of all elements in $W$ (and halts),
\item $C(w\vert W) \ge \log |W| - \beta$, which means that $w$ is essentially indistinguishable from many other  elements of $W$, and its membership in $W$ is its only significant property.
\end{itemize}
Not all strings are stochastic: provided that $\alpha$ and $\beta$ are much smaller than $n$, there exists an $n$-bit string $w$ that is not $(\alpha,\beta)$-stochastic (see \cite{shen-non-stochastic} and also \cite{vyugin-1985,vyugin-non-stochastic,vyugin-most-stochastic}).
Nevertheless, most objects arising naturally in mathematics and computer science satisfy the definition of $(\alpha,\beta)$-stochasticity for small parameters $\alpha,\beta$
(often of order $O(\log n)$).
V'yugin presented strong evidences that  ``\emph{data sequences normally occurring in the real world are stochastic},''  \cite{vyugin-most-stochastic}.
Moreover, we believe that in most relevant applications the set $W$ in the definition of stochasticity can be made \emph{highly explicit}
(in particular, membership in $W$ can be verified in polynomial time).

Revisiting the proof in~\cite{razenshteyn}, we make the following observations.
On one hand, the pairs $(x,y)$ constructed there do satisfy Kolmogorov's definition of $(O(\log n),O(\log n))$-stochasticity.
On the other hand, the set $S_n$ produced in the proof is by no means explicit.
We believe that membership in $S_n$ cannot be tested in less than double-exponential time.
Furthermore, for a typical pair $(x,y)\in S_n$, we expect a significant gap between the plain Kolmogorov complexities
\[
C(x), \ C(y),\ C(x,y),\ C(x|y),\ C(y|x)
\]
and the corresponding values of Levin’s $Kt$-complexity (see \cite{levin-kt} and \cite[Section~14]{vyugin-survey2}).
In other words, even though these objects admit short descriptions, the corresponding programs run for extremely long time, which indicates that they are somewhat exotic objects.
This lack of structure suggests that the construction in \cite{razenshteyn} is not well-suited for applications in areas such as communication complexity or coding theory.

In light of V'yugin’s perspective, we arrive at a stronger version of Question~\ref{q:2}:
does there exist an \emph{explicit and natural} example of a triple $(a,b,c)$  satisfying the requirements~\eqref{eq:halving}?
A related question:
does there exist an \emph{explicit and natural} pair $(x,y)$ satisfying conditions (i) and (ii) stated above?

The main result of this paper gives a positive answer to this stronger version of Question~\ref{q:2}.
Of course, the notion of explicitness and naturalness is somewhat subjective.
Nevertheless, we believe  that our construction has a clear algebraic and geometric interpretation and meets the standard criteria of explicitness commonly accepted in the AIT community.
In what follows, we describe this construction in more detail.

\subsection{The main construction}

Thus, the main goal of this paper is to provide a more explicit construction of a tuple that yields a negative answer to Question~\ref{q:2}.
We propose such a construction based on \emph{incidences in a finite projective plane}. 
We fix a finite field $\mathbb{F}$, take the projective plane over this field, and consider pairs $(x,y)$, where $x$ is a line in this plane and $y$ is a projective line passing through a point.
We call such pairs \emph{incidences}. An incidence in a projective plane is a classical combinatorial object, and its properties were extensively studied in different contexts.
Incidences were already studied in the context of AIT, see, e.g., \cite{muchnik-common-info,cmrsv}.

In a projective plane over $\mathbb{F}$ there are $\Theta(|\mathbb{F}|^2)$ points, $\Theta(|\mathbb{F}|^2)$ lines, and $\Theta(|\mathbb{F}|^3)$ incidences.
For the vast majority of incidences $(x,y)$ we have
\begin{equation}
\label{eq:typical}
C(x) \approx 2 \log |\mathbb{F}|, \ C(y) \approx 2 \log |\mathbb{F}|, C(x,y) \approx 3 \log |\mathbb{F}|.
\end{equation}
The upper bounds are trivial:   to specify a point or a line in a projective plane, it is enough to provide two elements of $\mathbb{F}$;
to specify together a point and a line incident to this point, it is enough to provide  three  elements of $\mathbb{F}$. 
The lower bound follows from a simple counting argument: the number of programs (descriptions) shorter than $k$ is less than $2^k$; therefore, for most incidences $(x,y)$
there is no short description, and $C(x,y) \approx 3 \log |\mathbb{F}|$. A similar argument implies $C(x) \approx 2 \log |\mathbb{F}|$  and $C(y) \approx 2 \log |\mathbb{F}|$.
We call an incidence $(x,y)$ \emph{typical} if it satisfies  \eqref{eq:typical}.

Thus, for a typical incidence $(x,y)$ the mutual information between $x$ and $y$ is
\[
I(x:y) := C(x) + C(y) - C(x,y) \approx  \log |\mathbb{F}|.
\]
An.~Muchnik observed in \cite{muchnik-common-info} that the mutual information of an incidence is hard to ``materialize,''  i.e., 
we cannot find a $z$ that ``embodies'' this amount of information shared by $x$ and $y$.
More formally, Muchnik proved that 
\begin{equation}
\nonumber
%\label{eq:common-info}
\text{there is no } z \text{ such that }C(z\vert x) \approx 0,\ C(z\vert y) \approx 0, \ C(z) \approx I(x:y).
\end{equation}
This insight did not close the question completely: the optimal  trade-off between $C(z\vert x)$, $C(z\vert y)$, $C(z)$ is still not fully understood. 
 Our work follows this direction of research. 
  We prove that for \emph{prime fields} $\mathbb{F}$, some specific values of $C(z\vert x)$, $C(z\vert x)$, and $C(z)$ are forbidden:
 \begin{equation}
\label{eq:main-result}
 \begin{array}{l}
 \text{For a prime $\mathbb{F}$, for a typical incidence $(x,y)$  there is no $z$ such that} \\
 C(z) \approx 1.5 \log |\mathbb{F}|, \ C(z\vert x) \approx 0.5\log |\mathbb{F}|, \  C(z\vert y) \approx 0.5\log |\mathbb{F}|, \  C(z\vert x,y) \approx 0, 
 \end{array}
 \end{equation}
see a more precise statement in Theorem~\ref{th:prime-field} on p.~\pageref{th:prime-field}. 
This result contrasts with a much simpler fact proven in \cite{shen-romashchenko}:
   \begin{equation}
\label{eq:anti-main-result}
 \begin{array}{l}
 \text{If  $\mathbb{F}$ contains a subfield of size  $\sqrt{|\mathbb{F}|}$, then 
  for a typical incidence $(x,y)$  there exists} \\
 \text{a } z \text{ such that }  
 C(z) \approx 1.5 \log |\mathbb{F}|, \ C(z\vert x) \approx  C(z \vert y) \approx 0.5\log |\mathbb{F}|, \  C(z\vert x,y) \approx 0,
 \end{array}
 \end{equation}
 see Theorem~\ref{th:cmrsv} on p.~\pageref{th:cmrsv}.
  
 Our proof of \eqref{eq:main-result} uses a remarkable result by Sophie Stevens and Frank De Zeeuw, 
 which gives a non-trivial upper bound on the number of incidences between points and lines in a plane over a prime field~\cite{tao_improved}.
The first theorem of this type was proven by Bourgain, Katz, and Tao, \cite{tao}. 
This result has been improved further in \cite{tao+1,tao+2,tao+3}.
We use the bound  from \cite{tao_improved}, the strongest to date.

Typical incidences in the projective plane over a prime field imply the negative answer to Question~\ref{q:2}.
We can prove this by following the idea sketched in the previous section: if $(x, y)$ is a typical incidence, we let
\[
a:=x,\ b:=y,\ c:=\langle x,y\rangle
\]
and  show that for every string $z$ satisfying the conditions 
\(
C(a\vert z) \approx \frac12 C(a) \text{ and }C(b\vert z) \approx \frac12 C(b),
\)
the value of $C(c\vert z)$ must be much smaller than $\frac12 C(c)$,  
see Corollary~\ref{th:shen's-question}.

\subsection{Application: impossibility results for secret key agreement}

The main result \eqref{eq:main-result} can be interpreted as a partial (very limited in scope)  answer to Question~\ref{q:3},  as it claims that for some specific pairs $(x,y)$
(typical incidences) there exist limitations for realizable complexity profiles of triples $(x,y,z)$. 
It is no surprise that  this fact can be used to prove certain \emph{no-go} results in communication complexity, for settings where the participants of the protocol are given 
such $x$ and $y$ as their inputs.
We present an example of such result --- a theorem on secret key agreement protocols, as we explain below. 

Unconditional \emph{secret key agreement} is one of the basic primitive in information-theoretic cryptography, \cite{unconditional-cryptography}.
In the simplest setting, this is a protocol for two parties, Alice and Bob. 
At the beginning of the communication, Alice and Bob are given some input data, $x$ and $y$ respectively.
It is assumed that $x$ and $y$ are strongly correlated, i.e., the mutual information between $x$ and $y$ is non-negligible.
Further, Alice and Bob exchange messages over a public channel and obtain (on both sides) some string $w$ that is incompressible (i.e., $C(w)$ is close to its length)
and has negligible mutual information with the transcript of the protocol, i.e.,
\[
C(w\vert \text{concatenation of all messages sent by Alice and Bob}) \approx |w|.
\]
Thus, Alice and Bob transform the mutual information between $x$ and $y$ into a common secret key 
(which can later be used, for example, into a one-time-pad or some other unconditionally secure cryptographic scheme). 
The \emph{secrecy} of the key  means that 
an eavesdropper should get (virtually) no information about this key, even having intercepted all communication between Alice and Bob. 
For a more detailed discussion of the secret key agreement  in the framework of AIT we refer the reader to \cite{andrei-marius,emirhan}.
\begin{remark}
In this paper, we assume that the communication protocol is \emph{uniformly computable};
that is, Alice and Bob exchange messages and compute the final result according to a single algorithmically defined rule that applies uniformly to inputs of all lengths.
We also assume that the protocol is public (i.e., known to an eavesdropper), so no secret information can be hardwired into the protocol description;
see \cite[Remark 1]{emirhan} and \cite[Remarks 2, 4, 13]{andrei-marius} for a more detailed discussion of the communication model.
\end{remark}
The challenges in secret key agreement are to (i) maximize the size of the secret key and (ii) to minimize the communication complexity of the protocol
(the total length of messages sent to each other by Alice and Bob). It is known that 
the maximum size of the secret key is equal to the mutual information between $x$ and $y$, i.e.,
\(
I(x:y) = C(x) + C(y) -C(x,y)
\)
(see  \cite{andrei-marius} for the proof in the framework of AIT  and \cite{secret-key-1,secret-key-2} for the original result in the classical Shannon's settings).
There exists a communication protocol that allows to produce a secret of optimal size with communication complexity  
\begin{equation}
\label{eq:cc}
\max\{C(x\vert y), C(y\vert x)\},
\end{equation}
 see \cite{andrei-marius}, and this communication complexity is tight, at least  for some ``hard'' pairs of inputs $(x,y)$, see \cite{emirhan}. 
 Moreover, subtler facts are known:
 \begin{itemize}
 \item the standard protocol achieving \eqref{eq:cc} (the construction dates back to \cite{secret-key-1,secret-key-2}; 
 see \cite{andrei-marius} for the AIT version) 
 is highly  asymmetric: all  messages are sent by only one party (Alice or Bob); 
 \item for some pairs of inputs $(x,y)$, if we want to agree on a secret key of maximal possible size $I(x:y)$, 
  not only the  total communication complexity must be equal to \eqref{eq:cc}, but actually 
 \emph{one of the parties} (Alice or Bob) must send 
$ \max\{C(x\vert y), C(y\vert x)\}$
% \eqref{eq:cc} 
bits of information, \cite{geoffroy-rustam};
 \item  for some pairs of inputs $(x,y)$, the total communication complexity 
 $\max\{C(x\vert y), C(y\vert x)\}$
% \eqref{eq:cc} 
cannot be reduced 
 \emph{even if the parties  need to agree on a sub-optimal secret key of size} $\delta n$ (for any constant $\delta>0$), see \cite{emirhan}.
 \end{itemize} 
 It remains unknown whether we can always organize a protocol of secret key agreement where the communication complexity \eqref{eq:cc} is shared evenly by the parties 
 (both Alice and Bob send $\frac12C(x\vert y) $ bits) if they need to agree on a key of sub-optimal size, e.g., $\frac12 I(x:y)$.

When we claim that communication complexity of a protocol is large \emph{in the worst case}, i.e., 
 Alice and Bob must send to each other quite a lot of bits \emph{at least for some pairs of inputs},
it is enough to prove this statement of some specific pair of data sets $(x,y)$. 
Such a  proof may become simpler when we use $(x,y)$ with nice combinatorial properties, even though these inputs may look artificial and unusual for practical applications.
Such is the case with the mentioned lower bounds for communication complexity  proven in  \cite{emirhan} and \cite{geoffroy-rustam}.
Both these arguments employ as an instance of a ``hard'' input $(x,y)$  a typical incidence in a finite projective plane.
Thus, it is natural to ask whether, for these specific $(x,y)$, it is possible to agree on a secret key of sub-optimal size using a \emph{balanced} 
communication load --- that is, when Alice and Bob each send approximately the same number of bits, roughly half the total communication complexity.
We show, quite surprisingly, that the answer to this question depends on whether the field admits a proper subfield:

\begin{positive-result*}
If the field $\mathbb{F}_q$ contains a subfield of size $\sqrt{q}$, then there exists a \emph{balanced} communication protocol with communication complexity $\log q$
where 
\begin{itemize}
\item Alice sends to Bob $\approx 0.5 \log q$ bits,
\item Bob sends to Alice $\approx 0.5 \log q$ bits,
\end{itemize}
and the parties agree on a secret key of length $\approx 0.5 \log q$, which is incompressible 
even conditional on the transcript of the communication between Alice and Bob.
\end{positive-result*}

\begin{negative-result*}
If the field $\mathbb{F}_q$ is prime, then in every \emph{balanced} communication protocol  with communication complexity $\log q$
such that 
\begin{itemize}
\item Alice sends to Bob $\approx 0.5 \log q$ bits,
\item Bob sends to Alice $\approx 0.5 \log q$ bits,
\end{itemize}
 the parties \emph{cannot} agree on a secret key of length $  \approx  0.5\log q$ 
 or even of any length $> \frac37 \log q$ 
 (the secrecy of the key means that the key must remain incompressible
even conditional on the transcript of the communication between Alice and Bob).
\end{negative-result*}
For a more precise statements see Theorem~\ref{th:ska-positive} and Theorem~\ref{th:ska-negative} respectively.

\subsection{Techniques}

A projective plane is a classical geometric object, and combinatorial properties of discrete projective planes have been studied with a large variety of mathematical techniques.
It is no surprise that, in the context of AIT, the information-theoretic properties of incidences in discrete projective planes have been studied using many different  mathematical tools.
In this paper we bring to AIT another (rather recent) mathematical technique that helps distinguish  information-theoretic properties of projective planes over
prime fields and over fields  containing proper subfields.

As we mentioned above, we apply 
the new approach to  the problem of  secret key agreement:
we consider  the setting where Alice and Bob receive as  inputs data sets $x$ and $y$  such that $(x,y)$ is a ``typical''
incidence in a projective plane ($x$ is a line and $y$ is a point incident to this line) 
over a finite field $\mathbb{F}$ with $n = \lceil \log  |\mathbb{F}| \rceil$. 
We summarize in Table~\ref{tab:1} below several technical results concerning this communication problem, and the techniques in the core of these results.

\begin{table}[h]
\begin{tabular}{| l | l |}
\hline
for any protocol of secret key agreement, & information-theoretic techniques:  \\
the size of the secret key $\lesssim  I(x:y) \approx n$ \cite{andrei-marius}  & $\text{internal informat.\,cost}\,{\le}\,\text{external informat.\,cost}$ \\ 
& (not specific for lines and points)\\
\hline
 $|\text{Alice's messages}| + |\text{Bob's messages}| \gtrapprox n$, & spectral method, expander mixing lemma \\
even for a secret key of size $\epsilon n$ \cite{emirhan}  &  (applies to all fast-mixing graphs, including\\ 
&the incidence graph of a projective plane)\\
\hline 
 $|\text{Alice's messages}|  \gtrapprox n$ or  $ |\text{Bob's messages}| \gtrapprox n$  if  & combinatorics of a projective plane\\
 the parties agree on a secret key of size  $ \approx n$, \cite{geoffroy-rustam} & (applies to all projective planes)\\
\hline
for incidences in a plane over a prime field if& additive combinatorics, algebraic and\\
 $|\text{Alice's messages}|  \approx 0.5 n$\;and\;$ |\text{Bob's messages}| \approx 0.5 n$ &geometric methods \cite{tao,tao+1,tao+2,tao+3,tao_improved} \\
  then the size of the  secret key   $\ll 0.5n$, \textbf{[this paper]}&(\textbf{applies to only projective planes} \\
 &\textbf{over prime fields}) \\
\hline
\end{tabular}
\caption{Bounds for secret key agreement in the framework of AIT}
\label{tab:1}
\end{table}
One of the motivations for writing this paper was to promote the notable results of \cite{tao,tao+1,tao+2,tao+3,tao_improved}, 
which presumably can find interesting applications in AIT and communication complexity.

\subsection{Organization} The rest of the paper is structured as follows. 
In Section~\ref{sec:halving} we briefly discuss \eqref{eq:anti-main-result}
(known from \cite{shen-romashchenko}).
In Section~\ref{sec:technical} we explain the (pretty standard) correspondence between information-theoretic and combinatorial properties of the incidences $(\text{line}, \text{point})$
on a descrete projective plane.
In Section~\ref{sec:no-halving} we formally prove our main result~\eqref{eq:main-result}.
In Section~\ref{sec:ska} we discuss an application of the main result: we show that the performance of the  secret key agreement 
for Alice and Bob  given as inputs an incident pair $(x, y)$ (from a projective plane) 
differs between fields that do and do not contain proper subfields.

\subsection{Notation}

\begin{itemize}
\item $|\cal S|$ stands for the cardinality of a finite set $\cal S$
\item we write $F(n) \ll G(n)$ if $G(n) - F(n) = \Omega(n)$ (e.g., $\frac{22n}{15} \ll \frac{3n}2$)
\item for a bit string $x$ we denote by $x_{k:m}$ a factor of $x$ that consists of $m-k+1$ bits at the positions between $k$ and $m$
(in particular, $x_{[1:m]}$ is a prefix of $x$ of length $m$);
\item we denote $\mathbb{FP}$ the projective plane over a finite field $\mathbb{F}$;
 \item $G=(R,L;E)$ stands for a bipartite graph where $L\cup R$ (disjoint union)
 is  the set of vertices and $E\subset L\times R$ is the set of edges;
\item $C(x)$ and $C(x\vert y)$ stand for Kolmogorov complexity of a string $x$ and, respectively, conditional Kolmogorov complexity  
of $x$ conditional on $y$, see \cite{li-vitanyi,suv}.
We use a similar notation for more involved expressions, e.g., 
$C(x,y\vert v,w)$ denotes Kolmogorov complexity of the code of the pair  $(x, y)$ conditional on the code of another pair $(v, w)$
\item we also talk about Kolmogorov complexity of more complex combinatorial objects (elements of finite fields, graphs, points and lines in a discrete projective plane, and so on)
assuming that each combinatorial object is represented by its \emph{code} 
(for some fixed computable encoding rule)
\item $I(x:y) := C(x)+C(y) - C(x,y)$ and $I(x:y\vert z) := C(x\vert z)+C(y\vert z) - C(x,y\vert z)$ stand for   information in $x$ on $y$ and, respectively,
 information in $x$ on $y$ conditional on $z$
\end{itemize}
Many natural equalities and inequalities for Kolmogorov complexity are valid only up to a logarithmic additive term, e.g., 
\(
C(x,y) = C(x) + C(y\vert x) \pm O(\log n),
\)
where $n$ is the sum of lengths of $x$ and $y$ (this is the chain rule a.k.a. Kolmogorov--Levin theorem, see \cite{zl}). To simplify the notation, we
write 
$
A\lep B
$
instead of 
\(
A \le B + O(\log N),
\)
where $N$ is the sum of lengths of all strings involved in the expressions $A$ and $B$. 
Similarly we define $A\gep B$ (which means $B\lep A$) and $A\eqp B$ (which means $A\lep B$ and $B\lep A$). 
For example, the chain rule can be expressed as
\(
C(x,y) \eqp C(x)+C(y\vert x).
\)

\section{Halving complexities of two strings}
\label{sec:halving}

In this section we discuss the positive answer to Question~\ref{q:2} for $k=1,2$ and $\lambda=1/2$. 
These results were proven in \cite{shen-romashchenko}. Here we recall the main ideas and technical tools behind this argument.

First of all, we observe that Question~\ref{q:2} for $k=1$ and $\lambda=1/2$ is pretty trivial. 
Given a string $x$ of length $N$, we can try $z=x_{[1:k]}$ for $k=0,\ldots N$. 
It is clear that for $k=0$ we have $C(x\vert x_{[1:k]}) = C(x)+O(1)$, and for $k=N$ we obtain $C(x\vert x_{[1:k]}) = O(1)$.
At the same time, when we add to the condition $z$ one bit, the conditional complexity $C(x\vert z)$ changes by only $O(1)$.
It follows immediately that for some intermediate value of $k$ we obtain  $z=x_{[1:k]}$ such that $C(x\vert z) = \frac12C(x) + O(1)$.

This argument employ (in a very naive from) the same intuition as the intermediate value theorem for continuous functions.
The case $k=2$ is more involved, but it also can be proven with ``topological'' considerations.

\begin{theorem}
\label{th:halving-log}
For all strings $a,b$ of complexity at most $n$ there exists a string $z$ such that 
\[
\left| C(a \vert z) - \frac12C(a) \right|  = O(\log n) \text{ and } \left| C(b \vert  z) - \frac12 C(b) \right| = O(\log n).
\]
\end{theorem}
In fact, \cite{shen-romashchenko} proved a tighter and more general  statement:
\begin{theorem} \cite[Theorem~4]{shen-romashchenko}
\label{th:halving-const}
For some constant $\kappa$ the following statement holds: for every two strings $a,b$ of
complexity at most $n$ and for every integers $\alpha,\beta$ such that
\begin{itemize}
\item $\alpha \le C(a) - \kappa \log n$,
\item $\beta \le C(b) - \kappa \log n$,
\item $- C(a \vert  b) + \kappa \log n \le \beta - \alpha \le C(b \vert  a) - \kappa\log n$,
\end{itemize}
there exists a string $z$ such that 
\(
| C(a \vert  z) - \alpha | \le  \kappa\text{ and }| C(b \vert  z) - \beta | \le \kappa.
\)

\end{theorem}
With $\alpha=\frac12C(a)$ and $\beta = \frac12 C(b)$, this theorem implies the following corollary,
which is (for non-degenerate parameters) a stronger version of Theorem~\ref{th:halving-log}:
\begin{corollary}
\label{c:halving-const}
For some constant $\kappa$ the following statement holds: for every two strings $a,b$ such that 
\(
C(a\vert b) \ge \kappa  \log n  \text{ and } C(b\vert a) \ge \kappa  \log n 
\)
there exists a string $z$ such that 
\[
\left| C(a \vert  z) - \frac12C(a) \right| \le  \kappa \text{ and } \left| C(b \vert  z) - \frac12 C(b) \right| \le \kappa.
\]
\end{corollary}
The proof of Theorem~\ref{th:halving-const}, due to A.~Shen, employs topological arguments. 
In outline, the construction of \( z \) proceeds by concatenating two parts: one derived from \( a \) and the other from \( b \). 
The principal difficulty is to choose the appropriate  sizes of these two components. 
It turns out that suitable proportions can indeed be chosen, and this fact follows from a well-known result in topology stated below.

\begin{proposition}
\label{p:topology}
Let \( D \) denote the two-dimensional disk with boundary circle \( S \). 
Suppose \( f : D \to \mathbb{R}^2 \) is a continuous map such that \( f(S) \subseteq S \). 
If the restriction \( f|_S \) has nonzero degree, then the image \( f(D) \) contains the entire disk \( D \); that is, every point of \( D \) has at least one preimage under \( f \).
\end{proposition}

\begin{remark}
A corresponding statement holds in higher dimensions. 
Let \( n \ge 1 \) be an integer, and let \( D^{n+1} \) denote the \((n+1)\)-dimensional disk with boundary sphere \( S^n \). 
Suppose \( f : D^{n+1} \to \mathbb{R}^{n+1} \) is a continuous map such that \( f(S^n) \subseteq S^n \). 
If the restriction \( f|_{S^n} \) has nonzero degree, then the image \( f(D^{n+1}) \) contains the entire disk \( D^{n+1} \).
\end{remark}

Proposition~\ref{p:topology} and its higher-dimensional generalization can be proven using the techniques presented, for example, in Chapter~26 of the textbook \cite{postnikov}. 
This argument also implies the classical topological result that a closed disk cannot be retracted onto its boundary circle (the so-called \emph{drum theorem}). 
Further discussion can be found in \cite{shen-romashchenko} and in the Appendix.

\section{Typical incidences  in a projective plane}
\label{sec:technical}

In  what follows we discuss typical pairs $(\text{line}, \text{point})$ in a finite projective plane and 
 their information-theoretic properties. 
The framework discussed in this section helps to translate information-theoretic questions in the combinatorial language.

\begin{definition}
\label{def-graph-parameters}
Let $G=(L,R;E)$ with $E\subset L\times R$ be a simple non-directed bipartite graph. This graph is bi-regular if all vertices in $L$ have the same degree (the same number of neighbors in $R$) 
and all vertices in $R$ have the same degree (the same number of neighbors in $L$).

To specify the quantitative characteristics of $G$ we will use a triple of parameters $(\alpha,\beta,\gamma)$ such that 
\[
|L| = 2^{\alpha}, \
|R| = 2^{\beta }, \
|E| = 2^{\gamma}. \
\]
If $G$ is bi-regular, then the degrees of vertices in $L$ are equal to $|E| / |L| = 2^{\gamma-\alpha}$ and the   degrees of vertices in $R$ are equal to 
$|E|/|R| = 2^{\gamma-\beta}$.
 \end{definition}
\begin{proposition}
\label{p:standard-profile}
Let $G = (L,R;E)$ be a bi-regular with parameters $(\alpha,\beta,\gamma)$, as defined above.
If the graph is given explicitly (the complete list of  vertices and edges of the graph can be found algorithmically given the value of the parameters $n$),
then for the vast majority (for instance, for $99\%$) of pairs $(x, y) \in E$ we have
\begin{equation}
\label{eq:profile}
C(x) \eqp \alpha,\ C(y) \eqp \beta,\ \text{and}\ C(x,y) \eqp  \gamma.
\end{equation}
\end{proposition}
\begin{proof}
This proposition follows from a standard counting, see e.g. \cite{suv}.
\end{proof}

\begin{definition}
\label{def-typical}
For a graph $G=(L,R;E)$ with parameters $(\alpha,\beta,\gamma)$ we  say that an edge $(u,v)\in E$ is \emph{typical} if it satisfies \eqref{eq:profile}.
\end{definition}
\begin{proposition}\label{p:complexity2subgraph}
Let $G = (L, R; E)$ be an explicitly given bi-regular bipartite graph with parameters $(\alpha, \beta, \gamma)$, as in Definition~\ref{def-graph-parameters}.
Let $(x, y) \in E$ be a typical edge in this graph, as in Definition~\ref{def-typical}. 
And let $z$ be a string satisfying:
\[
C(x \mid z) \le \alpha',\quad
C(y \mid z) \le \beta',\quad
C(x, y \mid z) \ge \gamma',
\]
for some integers $(\alpha', \beta', \gamma')$ with $\alpha' \le \alpha$, $\beta' \le \beta$, and $\gamma' \le \gamma$.
Then there exists an induced subgraph $H = (L', R'; E')$ of $G$,
\[
L' \subset L,\quad R' \subset R,\quad E' = (L' \times R') \cap E,
\]
such that
\(
|L'| = 2^{\alpha' \pm O(\log n)},\quad 
|R'| = 2^{\beta' \pm O(\log n)},\quad 
|E'| \ge 2^{\gamma' - O(\log n)}.
\)
\end{proposition}
\begin{proof}[Sketch of the proof]
We let
\[
L'= \{ x'\in L\ :\ C(x'\vert z) \le \alpha'\},
\
R' =  \{ y'\in R\ :\ C(y'\vert z) \le \beta' \}.
\]
Observe that $x\in L'$ and $y\in R'$.
\begin{lemma}
\label{l:complexity2combinatorics}
\(
|L'| = 2^{\alpha' \pm O(\log n)},\ 
|R'| = 2^{\beta'\pm  O(\log n)}.
\)
\end{lemma}
\begin{proof}[Proof of lemma]
This lemma is a standard translation between the combinatorial and the information-theoretic languages. 
 The upper bound for $|L'|$ follows from the fact that each element  of $L'$ is obtained from $z$ by a program of length at most 
 $\alpha'$. 
 The lower bound follows from the observation that  $L'$  contains, among other elements,  the $ 2^{\alpha' - O(\log n)}$
smallest elements of $L$ in lexicographic order. 
The argument for $R'$ is similar.
 A more detailed proof can be found, e.g., in  \cite[lemma~1 and lemma~2]{geoffroy-rustam}.
\end{proof}

It remains to prove a bound on the cardinality of $E'$. 
Given a string $z$, we can run in parallel all programs of length $\alpha'$ and $\beta'$ on input $z$ and enumerate the results that they produce. These results will provide us with the lists of elements $L'$ and $R'$ 
revealing step by step. Accordingly, we can enumerate edges of $E'$. Every pair $(x',y') \in E'$ can be specified by 
(i) the binary expansion of the numbers $\alpha',\beta'$ and (ii) by the ordinal number of $(x',y')$ in the enumeration of $E'$.
This argument applies in particular to the pair $(x,y)$, which belongs to $E'$.
 Therefore,
\(
C(x,y\vert z) \le \log |E'| + O(\log n).
\)
Reading this inequality from the right to the left, we obtain 
\[
|E'| \ge 2^{C(x,y |  z) - O(\log n)} = 2^{\gamma'-O(\log n)},
\] 
and we are done.
\end{proof}

\section{Graphs with highly non-extractible mutual information}
\label{sec:no-halving}

\subsection{A non-extractability result via random graphs}
\label{sec:no-halving-raz}
In this section, we revisit a result on the non-extractability of mutual information proven in~\cite{razenshteyn} with the technique of random graphs, 
and observe that it yields a negative answer to Question~\ref{q:2}  
(for \( k = 3 \) and \( \lambda = 1/2 \)).

\begin{theorem}[A special case of Theorem 9 in~\cite{razenshteyn}]
\label{th:raz}
For all sufficiently large \( n \), there exists a bipartite graph \( G = (L, R; E) \) with parameters \( (2n, 2n, 3n) \)  
such that for all subsets \( L' \subset L \) and \( R' \subset R \) of size \( |L'| = |R'| = 2^{n + o(n)} \),  
the number of edges in \( E \cap (L' \times R') \) is less than \( 2^{n + o(n)} \).
\end{theorem}

Applying Proposition~\ref{p:complexity2subgraph}, we obtain the following corollary.

\begin{corollary}
\label{cor:rasen}
For every \( \epsilon > 0 \) and all sufficiently large \( n \), there exists a pair \( (x, y) \) such that
\[
C(x) \eqp 2n, \quad C(y) \eqp 2n, \quad C(x, y) \eqp 3n,
\]
and for all strings \( z \) such that \( C(x \vert z) < (1 + \epsilon)n \) and \( C(y \vert z) < (1 + \epsilon)n \), we have
\begin{equation}
\label{eq:optimal-bound}
C(x, y \vert z) \lep (1 + O(\epsilon))n.
\end{equation}
\end{corollary}
This results gives an answer to Question~\ref{q:2}:
\begin{corollary}
\label{th:shen's-question-1}
For every \( n \), there exists a triple of strings \( (a, b, c) \), each of complexity \( \Theta(n) \), such that there is no string \( z \) satisfying
\[
\begin{array}{rcl}
C(a \vert z) &=& \frac{1}{2} C(a) + O(\log n), \\
C(b \vert z) &=& \frac{1}{2} C(b) + O(\log n), \\
C(c \vert z) &=& \frac{1}{2} C(c) + O(\log n).
\end{array}
\]
More precisely, for all \( z \) such that \( C(a \vert z) \eqp \frac{1}{2} C(a) \) and \( C(b \vert z) \eqp \frac{1}{2} C(b) \), we have
\[
C(c \vert z) \le n + o(n) \ll 1.5n \eqp \frac{1}{2} C(c).
\]
\end{corollary}

\begin{proof}
Fix an integer \( n \), and let \( (x, y) \) be the pair of strings from Corollary~\ref{cor:rasen}. We define
\[
a := x, \quad b := y, \quad c := \langle x, y \rangle
\]
and apply \eqref{eq:optimal-bound}.
\end{proof}

\begin{remark}
The proof of Theorem 9 in~\cite{razenshteyn} is non-constructive: the existence of the required graph \( G = (L, R; E) \) is established via the probabilistic method. 
Such a graph could, in principle, be found by brute-force search, but this would require doubly exponential time in \( n \). 
Indeed, one must enumerate all candidate bipartite graphs with \( 2^{2n} + 2^{2n} \) vertices and \( 2^{3n} \) edges, and for each graph, check all induced subgraphs with \( 2^n + 2^n \) vertices.

In what follows, we provide a much more explicit construction of a pair \( (x, y) \) with similar properties. 
These pairs will correspond to incidences in a projective plane over a prime field. 
We believe that such explicit examples of non-materializable mutual information may be more useful for applications, e.g., in communication complexity or coding theory 
(the first steps in this direction are presented in Section~\ref{sec:ska}).
However, the bounds we can prove for these explicit constructions are significantly weaker than~\eqref{eq:optimal-bound},
\end{remark}

\subsection{Typical incidences in a projective plane}

Now we instantiate the framework discussed above and discuss the central construction of this paper --- typical incidences in a finite projective plane.

 \begin{notation_nonumber}\label{e:projective-plane}
Let $\mathbb{F}$ be a finite field and $\mathbb{FP}$ be the projective plane over this field.
Let $L$ be the set of points and $R$ be the set of lines in this plane. A pair $(x,y)\in L\times R$ is connected by an edge iff the chosen  point $x$ lies in in the chosen line $y$.
Hereafter we denote this graph by $G^{\mathrm{PL}}_{\mathbb{F}}$.
\end{notation_nonumber}
We proceed with a discussion of properties of $(x,y)$ from $G^{\mathrm{PL}}_{\mathbb{F}}$ that differ depending on whether $\mathbb{F}$ possesses a proper subfield.

\begin{theorem}[see \cite{cmrsv}]
\label{th:cmrsv}
Let $\mathbb{F}$ be a field with a subfield of size $\sqrt{|\mathbb{F}|}$.
Then for a typical edge $(x,y)$ of  $G^{\mathrm{PL}}_{\mathbb{F}}$
(i.e., a typical incident pair $(\text{line},\text{point})$ on the plane $\mathbb{FP}$) we have 
\[
C(x) \eqp 2n ,\ 
C(y) \eqp 2n ,\ 
C(x,y) \eqp 3n ,
\]
and there exists a $z$ such that 
\[
C(x\vert z) \eqp n ,\ 
C(y\vert z) \eqp n ,\ 
C(x,y\vert z) \eqp 1.5n
\]
or, equivalently
\[
C(x\vert y,z) \eqp 0.5n ,\ 
C(y\vert x,z) \eqp 0.5n ,\ 
I(x:y\vert z) \eqp 0.5n,
\]
for $n=\log \lceil |\mathbb{F|}\rceil$, 
as shown in the diagram in Fig.~\ref{fig:lemma-halving-example}.
\end{theorem}
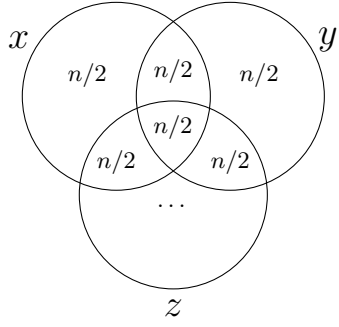
\begin{figure}[h]
\centering
				\begin{tikzpicture}[scale=0.5]
				  \draw \firstcircle node[above left] {\small $n/2$};
				  \draw \secondcircle node [above right] {\small $n/2$};
				  \draw \thirdcircle node [below] {\small \ldots};
				  \node at (95:0.05)   {\small $n/2$};
				  \node at (90:1.55) {\small $n/2$};
				  \node at (210:1.75) {\small $n/2$};
				  \node at (330:1.75) {\small $n/2$};
				  \node at (150:4.75) {\Large $x$};
				  \node at (30:4.75) {\Large $y$};
				  \node at (270:4.75) {\Large $z$};
				\end{tikzpicture}
\caption{Complexity profile for $(x,y,z)$ from Theorem~\ref{th:cmrsv}.}
\label{fig:lemma-halving-example}
\end{figure}
\begin{proof}[Sketch of the proof]
In what follows we sketch the scheme of the proof from  \cite{cmrsv}.
The first claim of the theorem, concerning the values of Kolmogorov complexities of $x$ and $y$, follows directly from Proposition~\ref{p:standard-profile} together with the fact that $(x,y)$ is typical.

The second claim, concerning $z$ and  the conditional Kolmogorov complexities of $x$ and $y$ given $z$, is more delicate and requires an appropriate construction. 
It can be reduced to the following combinatorial statement: the graph $G^{\mathrm{PL}}_{\mathbb{F}}$ can be covered by a relatively small family of induced subgraphs
\[
H_i = (L_i, R_i; E_i), \qquad 
|L_i| = |R_i| = 2^{n}, \quad |E_i| = 2^{1.5n},
\]
where in total $2^{1.5n + O(\log n)}$ such subgraphs are sufficient. 
Given a pair $(x,y)$, we select the subgraph $H_i$ that contains the edge $(x,y)$ and define $z$ as the index $i$ of this subgraph.
(If  $(x,y)$ is covered by more than one $H_i$, we may take any of these subgraphs.)

This covering property follows from two structural features of the graph of the projective plane  $G^{\mathrm{PL}}_{\mathbb{F}}$:  
(i)~the graph is edge-transitive, and  
(ii)~the field $\mathbb{F}$ contains a subfield $\mathbb{G}$ of cardinality $\sqrt{|\mathbb{F}|}$.

We first define a base subgraph $H_0$ as follows. 
Consider the subgraph induced by all vertices (points and lines in the projective plane) that can be represented by triples of field elements  taken from the subfield $\mathbb{G}$. 
Every other subgraph $H_i$ is then obtained from $H_0$ by the action of an appropriately chosen automorphism of the projective plane.

The automorphism group acts transitively on the set of incidences, that is, on all incident pairs $(\text{line}, \text{point})$. 
Therefore, when we apply a uniformly random automorphism to $H_0$, the probability that a fixed incidence $(\text{line}, \text{point})$ belongs to the image is
\[
\frac{\text{number of incidences in } H_0}{\text{number of all incidences in } G^{\mathrm{PL}}_{\mathbb{F}}}
=
\frac{\text{number of incidences in the plane over } \mathbb{G}}{\text{number of incidences in the plane over } \mathbb{F}} 
=
\frac{\Theta((\sqrt{q})^3)}{\Theta(q^3)}.
\]
Hence, if we  sample independently 
\[
\Theta\!\left((\sqrt{q})^3 \log q\right)
= 2^{(3/2)n + O(\log n)}
\]
automorphisms, then with high probability every incidence $(\text{line}, \text{point})$ in $G^{\mathrm{PL}}_{\mathbb{F}}$ is covered by at least one of the resulting subgraphs~$H_i$. 
We fix one such covering family $\{H_i\}$, thus completing the construction.

A more explicit method for constructing a covering family of subgraphs $H_i$ can be found in~\cite[Theorem~9]{cmrsv}.
See also Remark~\ref{remark-profile} after the proof of Theorem~\ref{th:ska-positive}.
\end{proof} 
Theorem~\ref{th:cmrsv} contrasts with Theorem~\ref{th:prime-field}.

\begin{theorem}\label{th:prime-field}
Let $\epsilon\ge 0$ be a small enough real number and 
 $\mathbb{F}$ be a field of a prime cardinality $p$,  and  $n := \lceil \log p \rceil$. Than for a typical edge $(x,y)$ of  $G^{\mathrm{PL}}_{\mathbb{F}}$
(i.e., a typical incident pair $(\text{line},\text{point})$ on the plane $\mathbb{FP}$) we have 
\[
C(x) \eqp 2n,\ 
C(y) \eqp 2n,\ 
C(x,y) \eqp 3n,
\]
and for every $z$ such that 
\begin{equation}
\label{eq:x|z_and_y|z}
C(x\vert z) \lep  (1+\epsilon)n,\ 
C(y\vert z) \lep  (1+\epsilon)n
\end{equation}
we have $C(x,y\vert z) \lep  %\frac{22}{15}(1+ \epsilon) n 
 (3/2-1/30+2\epsilon)n \ll 3n/2$.
\end{theorem}
\begin{proof}
Again, the first claim of the theorem (the values of unconditional Kolmogorov complexity) follows from Proposition~\ref{p:standard-profile} and from typicality of $(x,y)$.
We proceed with the second claim.
From Proposition~\ref{p:complexity2subgraph} it follows that in  $G^{\mathrm{PL}}_{\mathbb{F}}$
there is a subgraph  $G' = (L',R',E')$ such that 
\begin{equation}
\label{eq:2}
\begin{array}{rcl}
|L'| &=&  2^{(1+\epsilon)n+O(\log n)}, \\
|R'| &=&  2^{(1+\epsilon)n+O(\log n)},
\end{array}
\end{equation}
and 
\begin{equation}
\label{eq:1}
|E'| \ge 2^{C(x,y | z) - O(\log n) }.
\end{equation}
If $\epsilon<1/7$, then the cardinalities of $L'$ and $R'$ are less than $ |\mathbb{F}|^{8/7}$. 
It was shown in  \cite{tao_improved}  that for every subgraph $G'$ in  $G^{\mathrm{PL}}_{\mathbb{F}}$ for a prime $\mathbb{F}$ 
satisfying the constraints
\[
|L'|^{7/8} < |R'| < |L'|^{8/7} \text{ and } \max\{|L'|, |R'| \} \le |\mathbb{F}|^{8/7}  
\]
we have 
\[
|E'| \le (|L'| \cdot |R'|)^{11/15}.
\]
We plug in this inequality \eqref{eq:2} and \eqref{eq:1} and obtain 
\[
C(x,y\vert z) \lep \frac{22}{15}(1+ \epsilon) n   \le  (3/2-1/30+2\epsilon)n \ll 3n/2 , 
\]
provided that $\epsilon$ is small enough.
\end{proof}

Now we can use  prove a more constructive version of 
Corollary~\ref{th:shen's-question-1}.
\begin{corollary}
\label{th:shen's-question}
For every $n$ there exists a triple of strings $(a,b,c)$, each one of complexity $\Theta(n)$,  such that there is no $z$ satisfying 
\[
\begin{array}{rcl}
C(a\vert z) &=& \frac12 C(a) + O(\log n),\\
C(b\vert z) &=& \frac12 C(b) + O(\log n),\\
C(c\vert z) &=& \frac12 C(c) + O(\log n).
\end{array}
\] 
More precisely, for all $z$ such that $C(a\vert z) \eqp \frac12 C(a)$ and  $C(b\vert z) \eqp \frac12 C(b)$, we have  
\[
C(c\vert z) \le \frac{22}{45} C(c) + O(\log n )\ll \frac12 C(c).
\] 
Moreover, a suitable triple $(a,b,c)$ can be sampled from an explicitly given set tuples.
\end{corollary}
\begin{proof}
We fix an integer $n$  and the minimal prime number $p$ such that $2^n<p<2^{n+1}$,
and 
let $(x,y)$ be a typical edge in $G^{\mathrm{PL}}_{\mathbb{F}_p}$,  as in Theorem~\ref{th:prime-field}.
Then we define 
\(
a:= x,\ 
b:= y,\
c: = \langle x,y\rangle
\)
and apply Theorem~\ref{th:prime-field}.
\end{proof}

\section{Secret key agreement}
\label{sec:ska}

In this section we study communication complexity of the protocol of unconditional (in\-for\-mation-theoretic) secret key agreement.
Let us recall  the settings of the unconditional \emph{secret key agreement}.
We deal with two parties,  Alice and Bob. 
Alice and Bob receive input data, $x$ and $y$ respectively.
It is assumed that  the mutual information between $x$ and $y$ is non-negligible, and its value is known to Alice and Bob, as well as to the adversary.
Further, Alice and Bob exchange messages over a public channel and obtain (on both sides) some string $w$ that must be incompressible (i.e., $C(w)$ is close to its length)
and must have negligible mutual information with the transcript of the protocol, i.e.,
\[
C(w\vert \text{concatenation of all messages sent by Alice and Bob}) \approx |w|.
\]
Thus, Alice and Bob use  the mutual information between $x$ and $y$ to  produce a common secret key $w$
using a communication via a non-protected channel. 
The protocol succeed if Alice and Bob obtain one and the same $w$, and 
an eavesdropper gets only negligible information about this key, even having intercepted all messages sent to each other by Alice and Bob. 
In this paper we assume that the communication protocols are deterministic.
All arguments easily extends to randomized communication protocols with a public\footnote{%
The case of private sources of randomness is a more complex setting. We leave the consideration of this type of protocols for further research.%
} 
source of randomness
(accessible to Alice, Bob, and the eavesdropper).
A more detailed discussion of the settings of secret key agreement problem in the framework of AIT can be found in \cite{andrei-marius,emirhan}.

The optimal size of the secret key is known to be equal to the mutual information between $x$ and $y$, and 
 communication complexity of the protocol is at most \eqref{eq:cc}, see \cite{andrei-marius} 
(in what follows we discuss pairs $(x,y)$ with a symmetric complexity profile where $C(x\vert y) = C(y\vert x)$).

\subsection{Specific input data: secret key agreement with a typical incidence from a finite plane}
Let us focus on the case where the inputs $(x,y)$  represent a pair of typical incidences in a projective plane over a finite field $\mathbb{F}$ 
(we denote $n:=\lceil \log  | \mathbb{F} | \rceil$).
In this case the upper bound  \eqref{eq:cc} (which rewrites in this case to to $n$) is tight,  the communication complexity cannot be made better than $n-O(\log n)$, \cite{emirhan}. 
 Moreover, 
 \begin{itemize}
 \item[(i)] for every communication protocol, for its transcript $t$ we have 
 \begin{equation}
 \nonumber
 C(t) \gep I(t:x\vert y) + I(t: y \vert  x) \gep  n,
 \end{equation}
(the first inequality is known from \cite{andrei-marius} and the second one from  \cite{emirhan});
  \item[(ii)] this bound remains valid 
 \emph{even if the parties  agree on a sub-optimal secret key of size} $\delta n$ for any $\delta>0$,\cite{emirhan};
 \item[(iii)] if Alice and Bob agree on a secret key $w$ of maximal possible size $I(x:y) = n$,
then not only the  total communication complexity must be equal to $n$ but actually
 \emph{one of the parties} (Alice or Bob) must send 
 $\max\{C(x\vert y), C(y\vert x)\}\eqp n$ 
 bits of information, \cite{geoffroy-rustam}.
 \end{itemize} 
 We summarize: 
 \begin{itemize}
 \item \emph{even for a suboptimal key size} communication complexity of the protocol $\gep n$; 
\item  \emph{for an optimal key size}  the communication is very asymmetric --- all $n$ bits are sent by one of the participants.
\end{itemize}
 There remained a question:
 Does there exist a  protocol with a symmetric communication load (both Alice and Bob  send $\approx n/2$ bits)
with a suboptimal key size? In what follows we show that the answer to this question depends on whether the underlying field contains a proper subfield.

\subsection{Prime field: a negative result}

\begin{theorem}
\label{th:ska-negative}
Let $q$ be a prime number and $\mathbb{F}_q$  be  the field with $q$ elements. 
Let $\mathbb{FP}$ be the projective plane over  $\mathbb{F}_q$, and $(x,y)$ be a typical incidence in this plane
($x$ is a line in this projective plane, $y$ is a point in this line, and $C(x,y) \eqp 3 \log q$). 
Let us denote $n = \lceil \log q \rceil$.

We consider communication protocols where Alice is given as her input $x$ and Bob is given as his input $y$.
Assume that there exists a communication  protocol where
\begin{itemize}
\item Alice sends messages of total length $(\frac12+\epsilon)n$ bits to Bob, 
\item Bob  sends messages of total length $(\frac12+\epsilon)n$ bits to Alice,
\item at the end of the communication, Alice and Bob agree on a secret key $w$ of length $k$, satisfying
\(
C(w \vert  t ) \eqp C(w) \eqp k , 
\)
where  $t$ is the transcript of the protocol (the sequence of all messages exchanged between Alice and Bob during the protocol);
in other words,  the protocol reveals virtually no information about the secret to the eavesdropper.
\end{itemize}
We claim that 
for small enough $\epsilon$  the size of the secret key is much less than $\frac12 I(x:y)$, i.e., $k \ll n/2$. 
\end{theorem}
\begin{figure}[h]
\centering
				\begin{tikzpicture}[scale=0.65]
				  \draw \firstcircle node[above left]{\footnotesize$(0.5-\epsilon)n$};
				  \draw \secondcircle node [above right]{\footnotesize$(0.5-\epsilon)n$};
				  \draw \thirdcircle node [below] {}; 
				  \node at (95:0.05)   {\small $k$};
				  \node at (90:1.55) {\small $n-k$};
				  \node at (219:1.65) {\footnotesize$(0.5+\epsilon)n$};
				  \node at (321:1.65) {\footnotesize$(0.5+\epsilon)n$};
				  \node at (150:4.75) {\Large $x$};
				  \node at (30:4.75) {\Large $y$};
				  \node at (295:4.50) {\Large $z$};
				\end{tikzpicture}
\caption{Complexity profile for $(x,y,z)$ from Theorem~\ref{th:ska-negative}, cf. Fig.~\ref{fig:lemma-halving-example}.}
\label{fig:profile-ska}
\end{figure}
\begin{proof}
Let Alice and Bob agree on a secret key $w$ in protocol with transcript $t$. 
The fact that both Alice and Bob compute $w$ at the end of the protocol means that $C(w\vert t,x)$ and $C(w\vert t,y)$ are negligibly small. 
Security of  the key means that $I(w:t)$ is negligible, i.e., the transcript divulges virtually no information about the key.
Keeping in mind these observations, we define $z= \langle t,w\rangle$. We have
\(
C(z\vert x,y) = O(\log n)
\)
(given both $x$ and $y$, we can simulate the protocol and compute the transcript and the key). 
We may assume  that $C(t) \gep n$ (otherwise the size of the key is negligibly small, \cite{emirhan}).
On the other hand, since  Alice and Bob each send at most $(0.5+\epsilon)n$ bits, 
we have $C(t) \lep (1+2\epsilon)n$ and, moreover, 
\(
C(t\vert x) \lep  (0.5+\epsilon)n\text{ and }   
C(t \vert y) \lep  (0.5+\epsilon)n. 
\)

 Kolmogorov complexity of $z = \langle t,w\rangle$ is equal to  $C(t) + C(w)$ 
(the mutual information between $w$ and $t$ is negligible since protocol reveals no information about the secret). However, conditional on $x$ and conditional on $y$,
Kolmogorov complexities of $z$ and $t$ are essentially the same 
(given the transcript $t$ and the input of one of the parties, we can obtain the secret key $w$ for free).
It follows that 
\[
\begin{array}{rcl}
C(x\vert z) &\eqp& C(x,z) -  C(z)
 \eqp C(x) + C(z\vert x) -  C(z) \\
 &\eqp& C(x) + C(t\vert x) -  C(t,w) 
 \eqp C(x) + C(t\vert x) -  C(t) -  C(w) \\
 &\lep&  2n + (0.5+\epsilon) n -  n - k \eqp (1.5 + \epsilon)n - k.
\end{array}
\]
Similarly we obtain  
\(
C(y\vert z)   \lep (1.5 + \epsilon)n - k 
\)
and 
\[
\begin{array}{rcl}
C(x,y\vert z)& \eqp &C(x,y,z) -  C(z) \eqp C(x,y) + C(z\vert x,y) -  C(t) -  C(w)  \\
&\gep& 3n+0 - (1+2\epsilon)n - k  \eqp (2-2\epsilon)n -k.
\end{array}
\]
If  we assume now that 
 $k = \frac{n}2 \pm  O(\epsilon n)$,
we obtain 
\[
C(x\vert z) \lep n + O(\epsilon n),\ C(y\vert z) \lep n + O(\epsilon n),\  
C(x,y\vert z) \gep 1.5n - O(\epsilon n), 
\]
which  for small enough $\epsilon$ contradicts  Theorem~\ref{th:prime-field}. 
\end{proof}
\begin{remark}
Theorem~\ref{th:ska-negative} states that, for the given setting, in a communication protocol in which each party sends approximately $\frac12C(x\vert y) + O(\epsilon n) = \frac12C(y\vert x)  + O(\epsilon n)  = n/2  + O(\epsilon n) $  bits of information,
the size of the secret key cannot attain $\frac12I(x:y) = n/2$. 
Our proof (application of Theorem~\ref{th:prime-field}) actually  implies a  stronger bound: the size of the key 
cannot be greater than $3n/7 + O(\epsilon n) \ll \frac12I(x:y) $. 
\end{remark}

\subsection{Field with a large subfield: a positive result}

\begin{theorem}
\label{th:ska-positive}
Let $\mathbb{F}_q$  be  a field with $q$ elements, and $q=p^2$ for some integer $p$
(e.g., $p$ is prime and $q$ is a square of this prime number, or $p=2^k$ and $q = 2^{2k}$).

Let $\mathbb{FP}$ be the projective plane over  $\mathbb{F}_q$, and $(x,y)$ be a typical incidence in this plane
($x$ is a line in this projective plane, $y$ is a point in this line, and $C(x,y) = 3 \log q \pm O(\log n)$). 
We consider communication protocols where Alice is given as her input $x$ and Bob is given as his input $y$.
We claim that there exists a  communication protocol where 
\begin{itemize}
\item Alice sends a message  $m_A$ of length $n/2 $ bits to Bob,
\item Bob  sends a message  $m_B$ of length $n/2$ bits to Alice,
\item then Alice and Bob compute a secret key $w$ of length $n/2$ such that 
\[
C(w \vert \langle m_A, m_B\rangle) \gep n/2 ,
\]
where  $n = \lceil \log q \rceil$, 
i.e., the protocol reveals virtually no information about the secret to the eavesdropper.
\end{itemize}
\end{theorem}
\begin{proof}
A line $x$ and a point $y$ in the projective plane $\mathbb{FP}$ can be specified by their projective coordinates $(x_0:x_1:x_2)$ and $(y_0:y_1:y_2)$ respectively. 
Without loss of generality, we  assume $x_0\not =0$ and $y_2\not=0$ and denote
\[
x_1' := x_1 / x_0,\ x_2' :=  - x_2 / x_0
\text{ and }
y_0' := y_0 / y_2,\ y_1' := y_1 / y_2.
\]
The incidence of $x$ and $y$ means that 
\(
x_0y_0 + x_1y_1 + x_2y_2 = 0,
\)
or equivalently 
\begin{equation}
\label{eq:incidence}
y_0' + x_1'y_1' - x_2' = 0.
\end{equation}
Since $q=p^2$, the field $\mathbb{F}_q$ contains a subfield $\mathbb{G}$ of size $p$, and there exists an element $\xi\in \mathbb{F}_q$ such that  
every element $\alpha \in \mathbb{F}_q$ can be represented as
\(
\alpha = a_0 + a_1 \cdot \xi
\)
for some $a_0,a_1\in \mathbb{G}$. So we may represent $x_i'$ and $y_i'$ as follows:
\[
x'_1 = f+ r\xi,\  \ y'_0 = g+t\xi,\ y'_1 =  h+ s \xi
\]
for some $f,g,h,r,s,t \in \mathbb{G}$. In this notation, \eqref{eq:incidence} rewrites to 
\begin{equation*}
%\nonumber
(g+t\xi)+ (f+ r\xi) (h+ s \xi) = x_2'.
\end{equation*}
It follows that 
\begin{equation}
\label{eq:incidence-reduced-2}
x_2' = g+fh + (t+ fs + hr)\xi + rs\xi^2
\end{equation}
(The value $\xi^2$ can be represented as $u+v\xi$ for some $u,v\in \mathbb{G}$, but we do not need to specify these parameters.)
Let us recall that Alice knows all parameters derived from $x$ (including $f,r,x_2'$), and Bob knows all parameters derived from $y$ (including $g,h,s,t$).
We use the following  protocol.

\paragraph*{Communication protocol}

\begin{enumerate}
\item[] {\bf Round 1:} {Bob sends to Alice the value $m_1:= s$ (this message consists of $\log |\mathbb{G}| = n/2$ bits of information)}
\item[] {\bf Round 2:} {Alice computes $m_2 := g+fh$ and sends it to Bob
(this message also consists of $\log |\mathbb{G}| = n/2$ bits of information)}
\item[] {\bf Post-processing:} 
{Both participants compute the value $f$ and take it as the final result
(the secret key, which also consists of $\log |\mathbb{G}| = n/2$ bits of information).}
\end{enumerate}

\begin{claim}
Alice has enough information to compute  $m_2$.
\end{claim}
\begin{proof}[Proof of claim]
Initially, Alice is given the values of $x_1'=f+r\xi$ and $x_2' = u'+v'\xi$, where $f,r,u',v'$ are elements of $\mathbb{G}$. When she receives from Bob $s$, she gets all information to compute $rs\xi^2 = u'' +v''\xi$ (for some $u'',v''\in\mathbb{G}$). 
From \eqref{eq:incidence-reduced-2} it follows that 
\(
g+fh = u' - u''.
\)
\end{proof}

Alice is given the secret key $f$ as a part of her input. Bob, however, needs to do some computation to get this value. 

\begin{claim} Bob has enough information to compute the final result $f$.
\end{claim}
\begin{proof}[Proof of claim]
Initially, Bob was given the values $g, t, h,s\in \mathbb{G}$.
Bob receives from Alice the value
\(
g + f h,
\)
which is another  element  
of the field $\mathbb{G}$. This allows him to compute $f$ as
\(
((g + f h) - g) \cdot h^{-1}.
\)
\end{proof}
It remains to show that we reveal no information to the eavesdropper. The adversary can intercept the messages $m_1 = s$ and $m_2 = g+fh$. We need to show that these messages give no information about the produced secret key:

\begin{claim}
 $I(f : \langle m_1,m_2\rangle) = O(\log n)$.
\end{claim}
\begin{proof}[Proof of claim]
To specify the incidence $(x,y)$, it is enough to provide the values $f,g,h,r,s,t$ in $\mathbb{G}$. 
Thus, we have
\[
\begin{array}{rcl}
C(x,y) &\eqp &  C(f,g,h,s,r,t )  \\
  &\lep& C(m_1) + C(m_2) + C(f,g,h,s,r,t \vert  m_1,m_2)  \\
 &\lep& C(m_1) + C(m_2) +
C(s\vert m_1,m_2) + C(f\vert m_1,m_2) + C(h) \\
&& {}   + C(g\vert m_1,m_2,f,h) +  C(r) + C(t)  \\
&\lep& C(m_1) + C(m_2)  + C(f \vert  m_1, m_2)  + C(h) + C(r)  + C(t) \\
&\lep&  5 \log  |\mathbb{G}|  + C(f \vert  m_1, m_2)  \\
&\eqp & \frac52 n  + C(f \vert  m_1, m_2) 
\end{array}
\]
(in this calculation, $C(s \vert  m_1,m_2)$ vanishes since $m_1=s$, and $ C(g \vert m_1,m_2,f,h) $  vanishes since we can  compute $g$  given $f,h$ and the value of $g+fh$).

Since the incidence $(x,y)$ is typical, i.e., $C(x,y) \eqp 3n$, we obtain $ C(f \vert  m_1, m_2) \gep \frac{n}2 $.
Thus, 
\[
C(f \vert m_1, m_2) \gep C(f),
\] 
and the claim is proven.
\end{proof}
\end{proof}
\begin{remark}
\label{remark-profile}
Let \( z \) denote the tuple consisting of Alice’s message, Bob’s message, and the resulting secret key in the protocol from Theorem~\ref{th:ska-positive}. 
It is straightforward to verify that the complexity profile of \( (x, y, z) \) has the form shown in Fig.~\ref{fig:lemma-halving-example}. 
So we obtain another proof of Theorem~\ref{th:cmrsv}. 
In fact, this version of the proof is analogous to the explicit construction given in \cite[Theorem~9]{cmrsv}.
\end{remark}

\section{Conclusion}
Our proof of Theorem~\ref{th:prime-field} relies on a combinatorial property of the incidence graph of a projective plane over a prime field 
(the bipartite graph \( G^{\mathrm{PL}}_{\mathbb{F}_q} \) defined on p.~\pageref{def-graph-parameters}, with $\Theta(q^2)$ vertices and $\Theta(q^3)$ edges). 
Let \( A \) and \( B \) be subsets of vertices taken from the left and right parts of the graph, respectively, that is, from the sets of lines and points of the projective plane. 
If the cardinalities of both \( A \) and \( B \) are comparable to \( q \), then the number of edges connecting vertices in \( A \) and \( B \) is significantly smaller than \( q^{3/2} \). 
More precisely, it is bounded by \(O\left( q^{\frac32 - \delta}\right) \) for some  explicitly given \( \delta > 0 \).
According to  \cite{tao_improved},  any number less than $1/30$ can serve as $\delta$. 
The optimality of this bound is unknown.

A natural refinement of this result would be to increase the value of \( \delta \). 
At present, it is not known how large \( \delta \) can be in this combinatorial statement for projective planes over prime fields. 
It would be of interest to investigate other explicit constructions of bipartite graphs with comparable structural parameters  
(a similar number of vertices and edges) and with analogous or even stronger bounds on the number of edges between arbitrary vertex sets \( A \) and \( B \). 
Theorem~\ref{th:raz} shows that for a randomly chosen graph with the same number of vertices and edges as \( G^{\mathrm{PL}}_{\mathbb{F}_q} \),  
a similar property holds with \( \delta \) arbitrarily close to \( 1/2 \).  
Thus, the gap between bounds for explicit and random (implicit) constructions remains substantial.

Another natural direction for generalization is to consider the case where the sets \( A \) and \( B \) have cardinalities different from those considered above.  
Any progress on these combinatorial questions would contribute to a better  understanding of the information-theoretic properties of pairs \( (x, y) \) sharing large mutual information,  
and, ultimately, to new insights into problems in communication complexity.

\smallskip

We also emphasize that Question~\ref{q:1} on p.~\pageref{eq:2} (see also~\cite{27open}) remains  open.

\paragraph{Acknowledgments.}
The author is grateful to Ilya Shkredov for drawing attention to the paper~\cite{tao_improved}.
The author sincerely thanks the anonymous reviewers of %MFCS~2025 
for their careful reading and insightful suggestions, which helped to improve the clarity and presentation of the paper.

%\newpage
\bibliographystyle{plainurl}
\bibliography{barriers2025-doi}

%\newpage

\appendix

\section{Halving the complexities of two strings}

\subsection{Halving complexities with a hashing argument}

\begin{proposition}[An. Muchnik,\cite{muchnik2002conditional}]
\label{p:muchnik}
Let $x$ and $y$ be arbitrary strings of length at most $n$. Then there
exists a string $r$ of length $C(x | y)$ such that
\begin{itemize}
\item $C(r | x) \eqp0$ and
\item $C(x | r, y) \eqp0$.
\end{itemize}
Informally, $r$ is a ``digital fingerprint'' of $x$ (having a negligibly small complexity conditional on $x$) that can be used as a ``nearly optimal'' description of $x$ conditional on $y$.
\end{proposition}
\begin{corollary}
\label{c:muchnik}
Let $x$ and $y$ be arbitrary strings of length at most $n$ and let $k$ be a number less than $C(x)$. Then there
exists a string $r$ of length  $k$ such that
\begin{itemize}
\item $C(r) \eqp k$,
\item $C(r | x) \eqp 0$, and
\item $I(r : y) \eqp \max\{ 0, k -  C(x|y)\} $.
\end{itemize}
\end{corollary}
\begin{proof}[Proof of  the corollary]
First of all, we apply Proposition~\ref{p:muchnik} and obtain a string $r'$  of length $C(x | y)$  such that 
 $C(r' | x) \eqp0$ and $C(x | r', y) \eqp0$ (i.e., $r'$ is a ``fingerprint'' of $x$ that can be used as a nearly optimal description of $x$ conditional on $y$).  
 The latter condition means that the gap between $C(x|y)$ and $C(x|r',y)$ is
 equal to $C(x|y) \eqp C(r')$.  Therefore, 
\[
I(r' : y) \eqp 0.
\]
If $k\le C(x|y)$, we define $r$ as the prefix of $r'$ of length $k$, and this completes the construction. 
Otherwise, we apply  Proposition~\ref{p:muchnik} again and obtain $r''$  of length $C(x)-C(x|y)$ such that
 $C(r'' | x) \eqp0$ and $C(x | r', r'') \eqp0$  (i.e., $r''$ is a ``fingerprint'' of $x$ that can be used as a nearly optimal description of $x$ conditional on $r'$). 
 Let $r'''$ be the prefix of $r''$ of length $k-C(x|y)$. 
  We set $r = \langle r',r'''\rangle$ (encoding of the pair of  strings). It is easy to see that
 \[
 C(r) \eqp C(r') + C(r''') \eqp k 
 \]
and 
\[
 C(r|y) \eqp C(r'|y)  \eqp C(x|y) , 
\] 
which implies $I(r:y) \eqp  k -  C(x|y) $.
\end{proof}

\begin{proof}[Proof of Theorem \ref{th:halving-log}]
Let us denote 
\(
\alpha = C(a|b), \ \beta = C(b|a), \ \gamma = I(a:b).
\)
In this notation we have 
\[
C(a) \eqp \alpha+\gamma , \ 
C(b) \eqp \beta + \gamma .
\]
We assume without loss of generality  that $\alpha \le \beta$.

We construct $z$ as a combination of two components, $p$ and $q$, where $p$ and $q$ serve as  suitable ``fingerprints'' of $a$ and $b$ respectively.
First of all, we apply Corollary~\ref{c:muchnik} and obtain $p$ such that 
\begin{itemize}
\item $C(p) = \frac{\alpha+\gamma}2$,
\item $C(p|a) \eqp0$ and, therefore, $I(p:a) = C(p) = \frac{\alpha+\gamma}2$,
\item $ I(p:b) \eqp  \max\{ 0,   \frac{\alpha+\gamma}2 - \alpha \}  =  \max\{ 0, \frac{\gamma - \alpha}2\} $.
\end{itemize} 

\smallskip
\noindent
\emph{Case~1:} assume that $\alpha\ge \gamma$. In this case we have $ I(p:b) \eqp  \max\{ 0, \frac{\gamma - \alpha}2\}  \eqp 0$.
Observe that  $ \alpha\ge \gamma$ implies  $\beta \ge \gamma$. 

We apply  Corollary~\ref{c:muchnik}  again and take $q$ such that
\begin{itemize}
\item $C(q) \eqp \frac{\beta+\gamma}2$,
\item $C(q|b) \eqp O(\log n)$ and, therefore, $I(q:b) \eqp C(q) \eqp \frac{\beta+\gamma}2$,
\item $I(q:a) \eqp \max\{ 0,   \frac{\beta+\gamma}2 - \beta \} \eqp 0$ and, therefore, $I(q:p) \lep I(q:a) \eqp 0$.
\end{itemize} 
We let $z = \langle p, q \rangle$ (encoding of the pair of strings) and obtain 
\[
\begin{array}{rcl}
C(a|z) &\eqp &C(a|p,q)  \eqp C(a,p,q) - C(p,q)\\
&\eqp& C(a,q) - C(p,q)
\eqp C(a) + C(q) - I(q:a) - C(p) -C(q)\\
&\eqp& C(a) - C(p) \eqp \alpha+ \gamma  - \frac{\alpha+\gamma}2 \eqp  \frac{\alpha+\gamma}2 \eqp \frac12C(a),  
\end{array}
\]
and 
\[
\begin{array}{rcl}
C(b|z) &\eqp &C(b|p,q)  \eqp C(b,p,q) - C(p,q)\\
&\eqp& C(b,p) - C(p,q)
\eqp C(b) + C(p) - I(p:b) - C(p) -C(q)\\
&\eqp& C(b) - C(q) \eqp \beta+ \gamma - \frac{\beta+\gamma}2 \eqp  \frac{\beta+\gamma}2 \eqp \frac12C(b).
\end{array}
\]

\smallskip
\noindent
\emph{Case~2:}  assume that $\alpha < \gamma$. In this case 
$ I(p:b) \eqp  \max\{ 0, \frac{\gamma - \alpha}2\}  \eqp \frac{\gamma - \alpha}2$.
We apply one more time Corollary~\ref{c:muchnik}  and take $q$ such that
\begin{itemize}
\item $C(q) \eqp \frac{\alpha+\beta}2$,
\item $C(q|b) \eqp 0$ and, therefore, $I(q:b) \eqp C(q) \eqp \frac{\alpha+\beta}2$,
\item $I(q:a) \eqp   \max\{ 0, \frac{\alpha+\beta}2 - \beta \}  \eqp 0 $  and, therefore, 
$I(q:p) \lep I(q:a)  \eqp0 $.
\end{itemize} 
Similarly to Case~1, we let $z = \langle p, q\rangle$  and obtain 
\[
\begin{array}{rcl}
C(a|z) \eqp C(a|p,q) \eqp C(a|p) \eqp C(a) - I(p:a)  \eqp \alpha+\gamma - \frac{\alpha+\gamma}2  \eqp \frac{\alpha+\gamma}2 \eqp \frac12C(a),
\end{array}
\] 
and 
\[
\begin{array}{rcl}
C(b|z) &\eqp &C(b|p,q)  \eqp C(b,p,q) - C(p,q)\\
&\eqp& C(b,p) - C(p,q)
\eqp C(b) + C(p) -I(p:b)  - C(p) -C(q)\\
&\eqp& C(b) - I(p:b) - C(q) \eqp \beta+ \gamma -  \frac{\gamma-\alpha}2  - \frac{\alpha+\beta}2 \eqp \frac{\beta+\gamma}2 \eqp \frac12C(b).
\end{array}
\]
\end{proof}

\subsection{Halving complexities with a topological argument}

In this section, we present an alternative (topological) proof for one special case of Theorem~\ref{th:halving-log}. 
This \emph{special case} serves as a running example to which we return repeatedly throughout the paper. 
It involves a pair $(x, y)$ with the complexity profile
\[
C(x) \approx C(y) \approx 2n \quad \text{and} \quad I(x:y) \approx n.
\]
The argument outlined below can be combined with Muchnik’s hashing method from the previous section and extended to a significantly more general case; see~\cite{shen-romashchenko} for details.

\begin{proposition}
\label{th:topological-simple}
Let $x,y$ be strings such that
\[
C(x)  = 2n + o(n), \ C(y)  = 2n + o(n), \ I(x:y)  = n + o(n).
\]
Then there exists a string $z$ such that 
\(
C(x|z)  \eqp n \text{ and } C(y|z)  \eqp n .
\)
\end{proposition}
\begin{proof} (The argument below is a simplified version of the more general proof presented in  \cite[Theorem~4]{shen-romashchenko}.)
We will show that the required $z$ can be constructed as a pair $a=\langle x',y'\rangle$, where $x'$ and $y'$ are prefixes of $x$ and $y$ respectively. We are going to prove that such prefixes of $x$ and $y$
can be chosen appropriately, although we cannot specify their lengths explicitly.
The construction of $z$ depends on two parameters: the length $\alpha$ of $x'$ (which is an integer between $0$ and $|x|$) and the length $\beta$ of $y'$ (which is an integer between $0$ and $|y|$). 
The possible values of these parameters can be thought of as a two-dimensional grid of points with integer coordinates, see the integer points inside the rectangle in Fig.~\ref{fig:topology}. 
For each integer pair $(\alpha,\beta)$ in this rectangle we get the corresponding prefixes  $x'  = x_{[1:\alpha]} $ and $y'  = y_{[1:\beta]} $, and accordingly the pair $z=\langle x',y'\rangle$. 
This $z$ gives us in turn a pair of values 
\(
(C(x|z), C(y|z)),
\)
which is a point in the picture  on the right in Fig.~\ref{fig:topology}.
Observe that the mapping
\begin{equation}
\label{eq:continuous-mapping}
(\alpha,\beta) \mapsto (C(x\, |\, \langle x_{[1:\alpha]},y_{[1:\beta]} \rangle), C(y\, |\, \langle x_{[1:\alpha]},y_{[1:\beta]} \rangle)
\end{equation}
satisfies the \emph{Lipschitz condition}: increasing $\alpha$ or $\beta$ by $1$, we add one bit to $x'$ or $y'$; this operation changes the values of $C(x|\langle x',y\rangle)$ and  $C(y|\langle x',y\rangle)$
by only $O(1)$ additive terms. 
In the rest of the proof, our aim is to show that there exists a pair $(\alpha,\beta)$ on the left in Fig.~\ref{fig:topology} that is mapped to the $O(1)$-neighborhood of the point $(n,n)$ on the right.

%\begin{center}
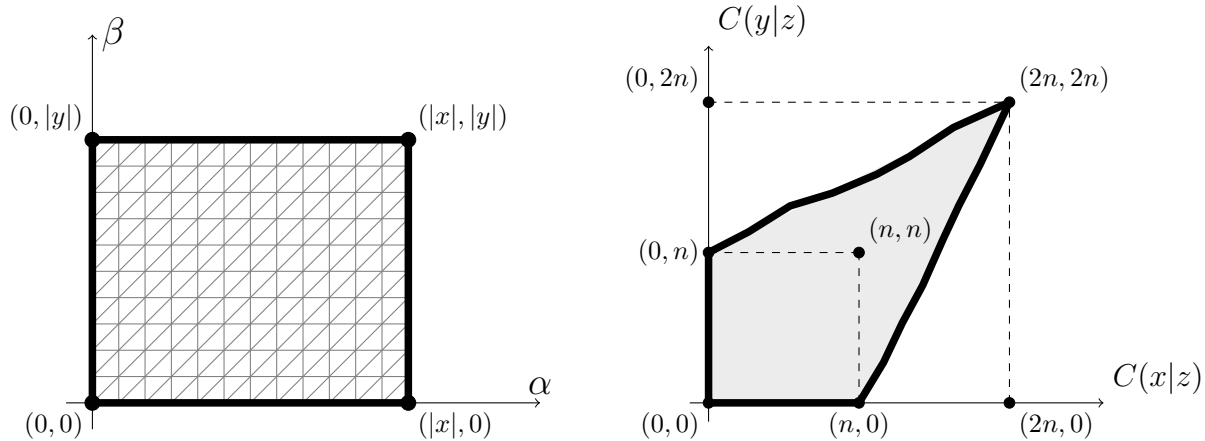
\begin{figure}
%\hspace{-15mm}
\begin{tikzpicture}[scale=0.35]

    % Coordinates for C(x) and C(y)
    \def\cx{12}  % C(x)
    \def\cy{10}  % C(y)

   % gray rectangle 
   \draw[fill=gray!15]  (0,0) -- (0,\cy) -- (\cx,\cy) -- (\cx,0)--cycle; 
  
    % Axes
    \draw[->] (-1,0) -- (17,0) node[above] {\Large $\alpha$};
    \draw[->] (0,-1) -- (0,14) node[right] {\Large $\beta$};

    % Number of divisions
    \def\Nx{12}
    \def\Ny{10}

    % Step size
    \pgfmathsetmacro{\dx}{\cx/\Nx}
    \pgfmathsetmacro{\dy}{\cy/\Ny}

    % Fill rectangle with dashed grid and diagonals
    \begin{scope}
        % Vertical grid lines
        \foreach \i in {1,...,11} {
            \draw[thin, gray] ({\i*\dx}, 0) -- ({\i*\dx}, \cy);
        }
        % Horizontal grid lines
        \foreach \j in {1,...,9} {
            \draw[thin, gray] (0, {\j*\dy}) -- (\cx, {\j*\dy});
        }

        % Diagonals in each small cell
        \foreach \i in {0,...,11} {
            \foreach \j in {0,...,9} {
                \draw[thin, gray] 
                    ({\i*\dx},{\j*\dy}) -- ({(\i+1)*\dx},{(\j+1)*\dy});
            }
        }
    \end{scope}
    
          \foreach \i in {0,...,12} {
                    \foreach \j in {0,...,10} {
                    \filldraw (\i,\j) circle (4pt) node {};
    		}
    	}

    % Rectangle border
    \draw[black, line width=0.7mm] (0,0) rectangle (\cx,\cy);

    % Points
    \filldraw (0,0) circle (8pt) node[below left] {$(0,0)$};
    \filldraw (\cx,0) circle (8pt) node[below right] {$(|x|,0)$};
    \filldraw (0,\cy) circle (8pt) node[above left] {$(0,|y|)$};
    \filldraw (\cx,\cy) circle (8pt) node[above right] {$(|x|,|y|)$};
    
\end{tikzpicture}
\hspace{5mm}
\begin{tikzpicture}[scale=0.25]
    % Constants
    \def\n{8}
   
    \pgfmathsetseed{2}

\tikzset{
        smooth wavy/.style={
            decorate, 
            decoration={snake, amplitude=1mm, segment length=6mm}
        },
        irregular/.style={
            decorate, 
            decoration={random steps, segment length=6pt, amplitude=1mm}
        }
    }

    % Draw light gray filled region first
    \begin{scope}
        \path[fill=gray!15]
            (0,0) --
            (0,\n) --
            (0,\n) -- ({2*\n},{2*\n})  --
             ({2*\n},{2*\n}) -- (\n,0) --
            (\n,0) -- (0,0) -- cycle;
    \end{scope}

      \draw[black, line width=0.7mm, decorate, decoration={random steps, amplitude=0.7mm, segment length=6mm}] 
        (\n,0) -- ({2*\n},{2*\n});  % from (n,0) to (2n,2n)     

    \draw[black, line width=0.7mm, decorate, decoration={random steps, amplitude=0.7mm, segment length=6mm}] 
        (0,\n) -- ({2*\n},{2*\n}); % from (0,n) to (2n,2n)

    \draw[line width=0.7mm] (0,0) -- (0,\n);    % Vertical from (0,0) to (0,n)
    \draw[line width=0.7mm] (0,0) -- (\n,0);    % Horizontal from (0,0) to (n,0)

    % Axes
    \draw[->] (-1,0) -- (21,0) node[above right] {\large $C(x|z)$};
    \draw[->] (0,-1) -- (0,19) node[above right] {\large $C(y|z)$};    
    
     \draw[dashed] (0,2*\n) -- (2*\n,2*\n);     
     \draw[dashed] (2*\n,0) -- (2*\n,2*\n);  
    
     \draw[dashed] (0,\n) -- (\n,\n);  
     \draw[dashed] (\n,0) -- (\n,\n);  
    
        % Points
    \filldraw (0,0) circle (8pt) node[below left] {$(0,0)$};
    \filldraw (0,{2*\n}) circle (8pt) node[above left] {$(0,2n)$};
    \filldraw ({2*\n},0) circle (8pt) node[below right] {$(2n,0)$};
    \filldraw ({2*\n},{2*\n}) circle (8pt) node[above right] {$(2n,2n)$};
    \filldraw (\n,\n) circle (8pt) node[above right] {$(n,n)$};
    
     \filldraw (\n,0) circle (8pt) node[below] {$(n,0)$};
     \filldraw (0,\n) circle (8pt) node[left] {$(0,n)$};
    
\end{tikzpicture}
\caption{Domain and codomain of the mapping  \eqref{eq:continuous-mapping}.}
\label{fig:topology}
\end{figure}
%\end{center}

We begin with the observation that the four points with coordinates 
\[
(0,0),\  (|x|,0),\  (0, |y|),\  (|x|,|y|)
\]
from the domain of the function \eqref{eq:continuous-mapping}  are
are mapped to the ($o(n)$-neighborhoods of) the points with coordinates 
\[
(2n,2n),\ (0,n), \ (n,0), \ (0,0)
\]
in the codomain\footnote{%
For example, $(\alpha,\beta) = (|x|,0)$ means that $z$ is a pair consisting of the entire $x$ and the empty prefix of $y$. Conditional on this $z$, complexity of $x$ vanishes, 
and conditional complexity of $y$ shrinks to $C(y|x) = n \pm o(n)$. The other three cases can be settled  in a similar way.%
}.

When we go along the vertical line $(|x|, 0)$--$(|x|,|y|)$ in the domain of \eqref{eq:continuous-mapping}, the string $z$ consists of the entire $x$ and an increasingly long prefix of $y$. Accordingly, the value of $C(x|z)$ remains equal to $O(1)$,
and $C(y|z)$ decreases from $C(y|z)\pm O(\log n)$ to $0$. That is, the image of $(\alpha,\beta)$ (taken from the domain of the function) moves along  the segment $(0,n)$--$(0,0)$ (in the codomain).
Similarly, when $(\alpha,\beta)$ runs the segment $(0,|y|)$--$(|x|,|y|)$ in the domain, the images run through the segment $(n,0)$--$(0,0)$ in the codomain, at least within $o(n)$ precision.

Now let $(\alpha,\beta)$ run through the segment $(0,0)$--$(|x|,0)$. We know that the image of $(\alpha,\beta)$ (i.e., the corresponding values of $(C(x|z), C(y|z))$) transitions  from
$(2n,2n)$ to $(0,n)$. However,  on this interval we cannot control exactly the behavior of $(C(x|z), C(y|z))$ since the image of $(0,0)$--$(|x|,0)$ does not need to be a straight line.
We only know that it is gradually descends from $(2n,2n)$ to $(0,n)$  and remains above the horizontal line $C(y|z) = n - o(n)$
(indeed, the complexity value $C(y|x')$ cannot be smaller than $C(y|x) = n - o(n)$).

Similarly, when $(\alpha,\beta)$ runs through the segment $(0,0)$--$(0, |y|)$, the corresponding pair of complexity values $(C(x|z), C(y|z))$  goes from $(2n,2n)$ to $(n,0)$. 
Thus, the boundary of the rectangle on the left in Fig.~\ref{fig:topology}
is mapped to the bold line on the right in the same figure. Observe that the point $(n,n)$ is inside the loop shown in bold on the right in Fig.~\ref{fig:topology}.
It remains to prove that some point from the domain of the mapping \eqref{eq:continuous-mapping} is mapped to the $O(1)$-neighborhood of $(n,n)$.  In what follows we do it in three steps.

\smallskip

\noindent
\textbf{Step 1: From discrete mapping to a continuous function.} The mapping  \eqref{eq:continuous-mapping} is defined only on the points with integer coordinates.
But we can extend it to the points with real coordinates (still inside of the rectangle $0\le \alpha \le |x|$, $0\le \beta \le |y|$. 
To this end, we split the grid of integer points into triangles (as shown in Fig.~\ref{fig:topology}); 
inside of each triangle, we define the function as the barycentric average of the values assigned to the vertices of this triangle.
The extended function is continuous and even Lipschitz continuous (as the original discrete function respected the condition of Lipschitz).

\smallskip
\noindent
\textbf{Step 2: Exact real-valued preimage of the point $(n,n)$.}  In our setting, we have a continuous mapping from a subset of $\mathbb{R}^2$ (the rectangle  $0\le \alpha \le |x|$, $0\le \beta \le |y|$, which is homeomorphic to the closed disk) to  $\mathbb{R}^2$.
We know that the boundary of the preimage is mapped to a curve (the bold contour on the right in Fig.~\ref{fig:topology}) that winds once around the point \((n,n)\).  
Moreover, when a point on the boundary of the preimage traverses the entire boundary once (the rectangle on the left in Fig.~\ref{fig:topology}),  
its image under the mapping moves along the contour on the right in Fig.~\ref{fig:topology}, likewise performing a complete revolution around the point \((n,n)\).  
It follows that there exists a point \((\alpha_0,\beta_0)\) in the preimage (possibly with non-integer coordinates) that is mapped exactly to \((n,n)\).  
This fact is a standard result from algebraic topology, following directly from Proposition~\ref{p:topology}  
(see, for instance, the discussion of the notion of the degree of a map and the proof of the \emph{drum theorem} in \cite[Chapter~26]{postnikov}).

\smallskip
\noindent
\textbf{Step 3: Approximate integer preimage of the point $(n,n)$.}  Since the mapping is Lipschitz continuous, we can replace
$(\alpha_0,\beta_0)$  by the closest point with integer coordinates $(\alpha_0',\beta_0')$. 
Obviously, this integer point is mapped by \eqref{eq:continuous-mapping}  to the $O(1)$-neighborhood of $(n,n)$. This observation completes the proof of the theorem.
\end{proof}

\section{Relativization does not change substantially the class of realizable complexity profiles}
\label{s:appendixB}

\begin{comment}
In this section, we explain why a positive answer to Question~\ref{q:2} would imply a positive answer to Question~\ref{q:1} (A.~Shen, personal communication~\cite{shen-private-communication}). 
Although we have shown that the general form of Question~\ref{q:2} has a negative answer, the argument introduces several interesting ideas. 
The proof relies on the method of \emph{typization}.

\begin{proposition}
\label{p:unrelativize}
For every tuple of strings $(y_1,\ldots, y_n)$ and every string $z$,  
there exists a tuple $(x_1,\ldots, x_n)$ such that for all index sets  
$1\le i_1 < i_2 < \ldots < i_s \le n$,
\[
C(x_{i_1},x_{i_2},\ldots,x_{i_s})
= C(y_{i_1},y_{i_2},\ldots,y_{i_s} \mid z)
\pm O(\log C(y_1,\ldots, y_n)).
\]
\end{proposition}
\end{comment}

In this section we prove Proposition~\ref{p:unrelativize}.

\begin{proof}
Let $N = C(y_1,\ldots,y_n)$.  
We begin with the standard \emph{typization trick} and define the  following set, which has  suitable sizes of projections and sections: 
\[
\begin{array}{rl}
S := \{%
&
(\tilde y_1,\ldots, \tilde y_n) \ : \
 \forall\ i_1<\ldots<i_s,\  
C(\tilde y_{i_1},\ldots,\tilde y_{i_s} \mid  z)
\le C(y_{i_1},\ldots,y_{i_s} \mid z),  \\
& \forall\ i_1<\ldots<i_s,\ \forall\ j_1<\ldots<j_m,\ 
C(\tilde y_{i_1},\ldots,\tilde y_{i_s} \mid \tilde y_{j_1},\ldots,\tilde y_{j_m}, z)
\le C(y_{i_1},\ldots,y_{i_s} \mid y_{j_1},\ldots,y_{j_m}, z)
\}.
\end{array}
\]

\begin{lemma}[see, e.g., Lemma~1 in~\cite{muchnik2010stability} or Theorem~211 in~\cite{suv}]
\label{l:typization}
The sizes of the projections and sections of $S$ correspond %(up to logarithmic precision) 
to the relevant Kolmogorov complexities of $(y_1,\ldots,y_n)$ conditional on $z$:
\begin{equation}
\label{eq:comb-entropy}
\left.
\begin{array}{l}
(i)\ \text{the cardinality of the entire set }S\text{ is } 2^{C(y_1,\ldots,y_n | z) \pm O(\log N)}, \\
(ii)\ \text{$\forall\ j_1<\ldots<j_m$ the cardinality of the projection of }S\text{ onto the coordinates $(j_1,\ldots,j_m)$}\\
\phantom{((ii)}\text{is at most } 2^{C(y_{j_1},\ldots,y_{j_m} | z) + 1}\\
(iii)\ \text{for every partition }\{1,\ldots,n\}  = \{i_1,\ldots,i_s\} \cup \{j_1,\ldots,j_{n-s}\} 
\text{ and for every } (\tilde y_{j_1},\ldots, \tilde y_{j_{n-s}}) \\ 
\phantom{(iii)}\text{there are at most}\ 
2^{C( y_{i_1},\ldots,  y_{i_s} | y_{j_1},\ldots, y_{j_{n-s}}, z)+ 1} \ 
\text{tuples}\  (\tilde y_{i_1},\ldots, \tilde y_{i_s})\ \text{such that} \\ 
\phantom{(iii)}\text{the combination of }\ 
(\tilde y_{i_1},\ldots, \tilde y_{i_s})\ 
\text{and}\  (\tilde y_{j_1},\ldots, \tilde y_{j_{n-s}})   \
\text{gives an $n$-tuple that belongs to}\
S
\end{array}
\right\}
\end{equation} 
\end{lemma}

\begin{proof}[Sketch of the proof]
The upper bound in (i), i.e., 
\(
S\le  2^{C(y_1,\ldots,y_n | z) + O(\log N)}
\) 
and even a slightly tighter
\(
S\le  2^{C(y_1,\ldots,y_n | z) +1}, 
\) 
follows by counting:  
the total number of tuples $(\tilde y_1,\ldots,\tilde y_n)$ satisfying
\[
C(\tilde y_{1},\ldots,\tilde y_{n} \mid  z)
\le C(y_{1},\ldots,y_{n} \mid z)
\]
cannot exceed the number of programs of length $\le C(y_1,\ldots,y_n \mid z)$.  
Analogous counting arguments yield (ii) and (iii).

For the lower bound in (i), note that $(y_1,\ldots,y_n)\in S$.  
Given $z$ and all values 
\(
C(y_{i_1},\ldots,y_{i_s} \mid z)
\)
and
\(C(y_{i_1},\ldots,y_{i_s} \mid y_{j_1},\ldots,y_{j_m}, z)\), 
one can run the enumeration of the list of all elements of $S$ and identify $(y_1,\ldots,y_n)$ by its ordinal number in this list. Hence,
\[
C(y_1,\ldots,y_n \mid z) \le \log |S| + O(\log N),
\]
implying $|S| \ge 2^{C(y_1,\ldots,y_n \mid z) - O(\log N)}$.
\end{proof}

We only use the existence of at least one set $S$ satisfying~\eqref{eq:comb-entropy}.  
Given all relevant complexity values involving $y_1,\ldots,y_n$ and $z$, we can, by brute force, construct \emph{some} set
\(
\hat S \subset (\{0,1\}^*)^n
\)
(possibly different from $S$) satisfying conditions (i)–(iii) of~\eqref{eq:comb-entropy}.  
Among all such sets, we select lexicographically the first one.  
Although $S$ itself may be highly complex and depend intricately on $z$, the procedure revealing $\hat S$ is simple (though slow).  
Hence the Kolmogorov complexity of the list of all elements of $\hat S$ is small, only $O(\log N)$, since we need to specify only the exponents appearing in~\eqref{eq:comb-entropy}
and do not need to know $z$.

For most tuples in $\hat S$, each component of their complexity profile, computed with plain Kolmogorov complexities, is close to the logarithm of the cardinality of the corresponding projection of $\hat S$.
More specifically, we have the following lemma.

\begin{lemma}[see, e.g., Lemmas~3--4 in~\cite{muchnik2010stability}]
\label{l:typization-2}
For most $(x_1,\ldots,x_n)\in \hat S$,
\[
\begin{array}{rcl}
\forall \ j_1\ldots j_m\ C(x_{j_1},\ldots x_{j_m})& \eqp &\log
\left[
\text{cardinality of the projection of}\
 \tilde S\
\text{onto the coordinates}\
({j_1},\ldots {j_m})
\right]
\end{array}
\]
\end{lemma}

\begin{proof}[Sketch of the proof]
The upper bound
\begin{equation}
\label{eq:appB-1}
C(x_{j_1},\ldots x_{j_m}) \lep \log
\left[
\text{cardinality of the projection of}\
 \tilde S\
\text{onto the coordinates}\
({j_1},\ldots {j_m})
\right]
\end{equation}
is immediate: the complexity of a tuple within an effectively defined set cannot exceed the logarithm of the set’s size.

For the lower bound, consider a partition
$\{1,\ldots,n\} = \{j_1,\ldots,j_m\} \cup \{i_1,\ldots,i_{n-m}\}$.  
If the left-hand side of~\eqref{eq:appB-1} were much smaller than the right-hand side, then by part (iii) of~\eqref{eq:comb-entropy} we would conclude that the value 
\[
C(x_1,\ldots,x_n)
\eqp
C(x_{j_1},\ldots,x_{j_m}) + 
C(x_{i_1},\ldots,x_{i_{n-m}} \vert x_{j_1},\ldots,x_{j_m})
\]
is much smaller than 
\[
C(y_{j_1},\ldots,y_{j_m} \vert z) + 
C(y_{i_1},\ldots,y_{i_{n-m}} \vert y_{j_1},\ldots,y_{j_m}, z) \eqp 
C(y_1,\ldots,y_n \mid z).
\] 
This cannot be true for most tuples in $\hat S$, because by part (i) of~\eqref{eq:comb-entropy} we have $\log|\hat S| \eqp C(y_1,\ldots,y_n \mid z)$.
\end{proof}

Applying Lemma~\ref{l:typization-2}, we obtain a tuple $(x_1,\ldots,x_n)$ such that for all index sets $\{j_1,\ldots,j_m\}\subset\{1,\ldots,n\}$,
\[
C(x_{j_1},\ldots,x_{j_m})
\eqp
C(y_{j_1},\ldots,y_{j_m}\mid z),
\]
which completes the proof of the proposition.
\end{proof}

Now we can explain how the positive answer to Question~\ref{q:2} would imply the positive answer to Question~\ref{q:1}.
We restrict our attention to the principal case of Question~\ref{q:1}, where $\lambda \in (0,1)$.
\begin{corollary}
\label{cor:q2->q1}
Let $\lambda\in(0,1)$. 
Assume that for every $k$-tuple of strings $(y_1,\ldots,y_k)$
there exists a string $z$ such that
\[
C(y_i \mid z) = \lambda\, C(y_i) + O(\log C(y_1,\ldots,y_k))
\quad\text{for all } i=1,\ldots,k.
\]
Then for every $n$ and every $n$-tuple of strings $(x_1,\ldots,x_n)$
there exists another $n$-tuple $(x_1',\ldots,x_n')$ such that
\[
C(x_{i_1}',x_{i_2}',\ldots,x_{i_s}')
= \lambda\, C(x_{i_1},x_{i_2},\ldots,x_{i_s})
+ O(\log C(x_1,\ldots,x_n))
\]
for all $1\le i_1<i_2<\ldots<i_s\le n$.
\end{corollary}

\begin{proof}
Fix $\lambda\in(0,1)$ and an $n$-tuple $(x_1,\ldots,x_n)$, and let $N=C(x_1,\ldots,x_n)$.  
Define $k=2^n-1$ strings $y_J$, indexed by all nonempty subsets $J\subseteq\{1,\ldots,n\}$, as
\[
y_J := \langle x_{j_1},\ldots,x_{j_m}\rangle
\quad\text{for } J=\{j_1,\ldots,j_m\}.
\]
By the assumption, there exists a string $z$ such that for every $J$ we have
\(
C(y_J \mid z) \eqp \lambda\, C(y_J),
\)
which is equivalent to
\[
C(x_{j_1},\ldots,x_{j_m} \mid z)
\eqp
\lambda\, C(x_{j_1},\ldots,x_{j_m}).
\]
By Proposition~\ref{p:unrelativize}, there exists a tuple $(x_1',\ldots,x_n')$ such that for all $j_1<\ldots<j_m$,
\[
C(x_{j_1}',\ldots,x_{j_m}')
\eqp
C(x_{j_1},\ldots,x_{j_m}\mid z)
\eqp
\lambda\, C(x_{j_1},\ldots,x_{j_m}),
\]
as required.
\end{proof}

Although we have shown that, in general, the answer to Question~\ref{q:2} is negative, the above observation may still be useful for certain special cases 
(for strings $x_j$ or values of $\lambda$ of a specific form, where the answer to Question~\ref{q:2} becomes positive).

\end{document}